\documentclass{llncs}

\usepackage{hyperref}
\usepackage{enumerate} 
\usepackage{url}
\usepackage[pdftex]{graphicx}
\usepackage{amsmath}
 \usepackage{amssymb}
\usepackage{amsfonts}
\usepackage{todonotes}
\usepackage{paralist}
\usepackage{multirow}
\usepackage{caption}
\usepackage{fancyvrb}
\usepackage{lscape}
\usepackage{listings}
\usepackage{pdflscape}
\usepackage{wrapfig}
\usepackage{tikz}
\usetikzlibrary{shapes}
\usetikzlibrary{arrows}
\usetikzlibrary{positioning}
\usetikzlibrary{automata}
\usetikzlibrary{calc}
\usepackage{subcaption}

\newcommand{\bigmathset}[2]{\big\{\: #1 \;|\; #2 \:\big\}}
\newcommand{\mset}[1]{\{ #1 \}}
\newcommand{\bigmset}[1]{\big\{ #1 \big\}}

\newcommand{\Ab}{\mathbb{A}}
\newcommand{\Zb}{\mathbb{Z}}
\newcommand{\Nb}{\mathbb{N}}

\newcommand{\ra}{\rightarrow}

\newcommand{\tuple}[1]{\langle #1 \rangle}

\newcommand{\asort}{\textsf{ASort}}
\newcommand{\isort}{\textsf{ISort}}
\newcommand{\bsort}{\textsf{BSort}}
\newcommand{\vsort}[1]{\textsf{VSort}^#1}
\newcommand{\fsort}[1]{\textsf{FSort}^#1}

\newcommand{\VA}{\var{Var}_A}
\newcommand{\VI}{\var{Var}_I}

\newcommand{\pp}{\texttt{++}}
\newcommand{\pe}{\texttt{+=}}
\newcommand{\mm}{\texttt{--}}

\newcommand{\fand}{\:\land\:}
\newcommand{\for}{\:\lor\:}

\newcommand{\Break}{\ensuremath{\mathit{break}}}
\newcommand{\Skip}{\ensuremath{\mathit{skip}}}

\newcommand{\Fold}{\ensuremath{\mathit{fold}}}

\newcommand{\FoldPI}[3]{\Fold_{#1}\!\bm{#2}\!\bm{#3}}

\newcommand{\model}{\sigma}
\newcommand{\true}{\var{true}}
\newcommand{\false}{\var{false}}

\newcommand{\guard}{\mathit{grd}}

\newcommand{\update}{\mathit{upd}}

\newcommand{\elm}{\mathbf{e}}
\newcommand{\idx}{\mathbf{i}}
\newcommand{\cnt}{\mathbf{c}}
\newcommand{\st}{\mathbf{s}}

\newcommand{\pen}{\,\pe\,n}

\newcommand{\ev}[1]{\left[ #1 \right]^\sigma}
\newcommand{\evk}[1]{\left[ #1 \right]^\kappa}
\newcommand{\evsk}[1]{\left[ #1 \right]^{\sigma,\kappa}}
\newcommand{\evskp}[1]{\left[ #1 \right]^{\sigma,\kappa'}}
\newcommand{\FV}{\var{FV}^m}

\newcommand{\NP}{\mathbf{NP}}

\newcommand{\PSPACE}{\mathbf{PSPACE}}

\newcommand{\scm}{\mathcal{M}}
\newcommand{\vars}{X}
\newcommand{\ctr}{\eta}
\newcommand{\ctrv}{\vec{\ctr}}
\newcommand{\states}{Q}
\newcommand{\alphabet}{\Sigma}
\newcommand{\cc}{\mathsf{CC}}
\newcommand{\ic}{\mathsf{IC}}
\newcommand{\init}{q^\mathrm{init}}
\newcommand{\tr}{\delta}

\newcommand{\conf}{\zeta}

\newcommand{\invar}{x_\mathbf{e}}

\newcount\colveccount
\newcommand*\colvec[1]{
        \global\colveccount#1
        \begin{pmatrix}
        \colvecnext
}
\def\colvecnext#1{
        #1
        \global\advance\colveccount-1
        \ifnum\colveccount>0
                \\
                \expandafter\colvecnext
        \else
                \end{pmatrix}
        \fi
}

\newcommand{\bm}[1]{\begin{pmatrix}#1\end{pmatrix}}

\makeatletter
\newcommand{\vect}[2][]{%
  \gdef\@VORNE{1}
  \hskip-\arraycolsep%
    \begin{smallmatrix}{#1}\vekSp@lten{#2}\end{smallmatrix}%
  \hskip-\arraycolsep}

\def\vekSp@lten#1{\xvekSp@lten#1;vekL@stLine;}
\def\vekL@stLine{vekL@stLine}
\def\xvekSp@lten#1;{\def\temp{#1}%
  \ifx\temp\vekL@stLine
  \else
    \ifnum\@VORNE=1\gdef\@VORNE{0}
    \else\@arraycr\fi%
    #1%
    \expandafter\xvekSp@lten
  \fi}
\makeatother

\newcommand{\Vect}[1]{\bm{\vect{#1}}}

\newcommand{\Then}{\;\Rightarrow\;}

\newcommand{\var}[1]{\ensuremath{\mathit{#1}}}

\begin{document}

%\title {Array Folds Logic} % \\for Symbolic Testing and Verification}
 \title{Array Folds Logic\thanks{This research was
     supported in part by the European Research Council (ERC) under
     grant 267989 (QUAREM) and by the Austrian Science Fund (FWF) under
     grants S11402-N23 (RiSE and SHiNE) and Z211-N23 (Wittgenstein Award).}  }

  \author{Przemys\l aw Daca \and
    Thomas A.\ Henzinger \and Andrey Kupriyanov} \institute{IST Austria, Austria}

\date{\today}
\maketitle

\begin{abstract}
%We present an extension to the existential theory of integer arrays
%which allows us to express counting and accumulation. The properties
%expressible in Array Folds Logic (AFL) include statements such as ``the 
%first array cell contains the array length,'' and ``the array contains
%equally many minimal and maximal elements.'' These properties 
%cannot be expressed in quantified fragments of the theory of arrays, 
%nor in the theory of concatenation. Using reduction to counter machines, 
%we show that the satisfiability problem of AFL is NEXP-complete, and 
%under a natural restriction the complexity decreases to NP. We also 
%show that adding either universal quantifiers or concatenation leads to 
%undecidability.

We present an extension to the quantifier-free theory of integer arrays which allows us to express counting. The properties expressible in Array Folds Logic (AFL) include statements such as ``the first array cell contains the array length,'' and ``the array contains equally many minimal and maximal elements.'' These properties cannot be expressed in quantified fragments of the theory of arrays, nor in the theory of concatenation. Using reduction to counter machines, we show that the satisfiability problem of AFL is PSPACE-complete, and with a natural restriction the complexity decreases to NP. We also show that adding either universal quantifiers or concatenation leads to undecidability.

AFL contains terms that fold a function over an array. We demonstrate that folding, a well-known concept from functional languages, allows us to concisely summarize loops that count over arrays, which occurs frequently in real-life programs. We provide a tool that can discharge proof obligations in AFL, and we demonstrate on practical examples that our decision procedure can solve a broad range of problems in symbolic testing and program verification.

\end{abstract}

\section{Introduction}
\label{sec:intro}

Arrays and lists (or, more generally, sequences) are fundamental data structures both for imperative and functional programs: hardly any real-life program can work without processing sequentially-ordered data. Testing and verification of array- and list-manipulating programs is thus a task of crucial importance. Almost any non-trivial property about these data structures requires some sort of universal quantification; unfortunately, the full first-order theories of arrays and lists are undecidable. This has motivated researches to investigate fragments with restricted quantifier prefixes, and has given rise to numerous logics that can describe interesting properties of sequences, such as partitioning or sortedness. These logics have efficient decision procedures and have been successfully applied to verify some important aspects of programs working with arrays and lists: for example, the correctness of sorting algorithms.

However, an important class of properties, namely, \emph{counting} 
%and \emph{accumulation} 
 over arrays, has eluded researchers' attention so far. 
In addition to the examples from the abstract, this includes statements such as 
``the histogram of the input data satisfies the given distribution,''
%``the sequence has the shape $a^n\:b^n\:c^n$ for some numbers $a<b<c<n$,'' 
or 
``the packet adheres to the requirements of the given type-length-value (TLV) encoding (e.g., of the IPv6 options).''
% or ``the average computed over the array is bounded.'' 
%``the average computed over the array is bounded.'' 
%This class includes, besides the statements shown in the abstract, such properties as ``the histogram of the input data satisfies the given distribution,''  or ``the average computed over the array is bounded.'' 
Such properties, though crucial for many applications, cannot be expressed in decidable fragments of the first-order theory of arrays, nor in the decidable extensions of the theory of concatenation. 
%These properties cannot be expressed in quantified fragments of the theory of arrays, nor in the theory of lists with concatenation. At the same time, this properties are crucial for many applications that concern with parsing, data analysis, or stream processing.

In this paper we present \emph{Array Folds Logic} (\emph{AFL}), which is an extension of the quantifier-free theory of integer arrays. But instead of introducing quantifiers, we introduce counting 
%and accumulation 
in the form of \Fold\ terms. Folding is a well-known concept in functional languages: as the name suggests, it folds some function over an array, i.e., applies it to every element of the array in sequence, while preserving the intermediate result. %Inside \Fold\ terms we allow a control structure with statements about counter variables.
%or accumulation 

%  Contrary to these logics, we do not introduce quantifiers, but a special construct, \emph{folding}, which is well-known from functional languages. 

To illustrate the kind of problems we are dealing with, consider the following toy example: given an array, accept it if the number of minimum elements in the array is the same as the number of maximum elements in the array. E.g., the array $[1,2,7,4,1,3,7,5]$ is accepted (because there are two $1$'s and two $7$'s), while the array $[1,2,7,4,1,3,6,5]$ is rejected (because there is only one $7$).

Written in a programming language like \texttt{C}, the problem can be solved by the piece of code shown in Figure~\ref{fig:toy-a}, but such \emph{explicit} solution cannot express verification conditions for \emph{symbolic} verification and testing.
%Clearly, this solves the problem for one  \emph{concrete} array, but cannot be integrated into a \emph{symbolic} verification or testing method. 
We can use the quantified theory of arrays mixed with assertions about cardinality of sets, as in Figure~\ref{fig:toy-b}. Unfortunately, such a combination is undecidable (by a reduction from Hilbert's Tenth Problem: replace \Fold{s} with cardinalities in the proof of Theorem~\ref{hilbert}).
%there is no logic and decision procedure capable of handling such a combination. 

\newcommand{\Min}{\var{min}}
\newcommand{\Max}{\var{max}}

\begin{figure}[t]
\centering
\begin{subfigure}[b]{.45\textwidth}
\begin{small}
\begin{verbatim}
min = max = a[0];
j = k = 0;
for(i=0;i<size(a);i++) {
  if(a[i]<min) { min=a[i]; j=1; }
  if(a[i]==min) j++;
}
for(i=0;i<size(a);i++) {
  if(a[i]>max) { max=a[i]; k=1; }
  if(a[i]==max) k++;
}
assert(j==k);
\end{verbatim}
\end{small}
\vspace*{-2mm}
\caption{\texttt{C} language}
\label{fig:toy-a}
\end{subfigure}
\hspace*{5mm}
\begin{subfigure}[b]{0.45\textwidth}
%\scalebox{0.8}{
\begin{align*}
& \exists\, \Min, \Max, i_1, i_2, j, k \;.\; \\
& 0 \le i_1 < |a| \fand  0 \le i_2 < |a| \fand \\
& a[i_1] = \Min \fand  a[i_2] = \Max \fand \\
& \forall i. (a[i] \ge \Min) \fand \\
& \forall i.(a[i] \le \Max)  \fand  \\
& j=\big|\{ i \:|\: a[i]=\Min \} \big| \;\land \\
& k=\big|\{ i \:|\: a[i]=\Max \} \big| \;\land \\
& j = k
\end{align*}
\vspace*{-2mm}
\caption{Quantified arrays + cardinality}
\label{fig:toy-b}
\end{subfigure}
\vspace*{2mm}

\begin{subfigure}[b]{\textwidth}
\begin{small}
\begin{gather*}
0 \le i_1 < |a| \fand  0 \le i_2 < |a| \fand a[i_1] = \Min \fand  a[i_2] = \Max \fand \\
\FoldPI{a}{\vect{0;0}}{\vect{ 
\elm=\Min \Then \cnt_1\pp; \elm > \Min \Then \Skip}} = \Vect{|a|;j} \fand \FoldPI{a}{\vect{0;0}}{\vect{\elm=\Max \Then \cnt_1\pp;\elm < \Max \Then \Skip }}=\Vect{|a|;k} \fand  j=k
\end{gather*}
%\begin{gather*}
%j = \FoldPI{a}{c=0}{\vect{a \ge \Min; i=i_1 \Then a = \Min; 
%a=\Min \;\Then c\pp}} \fand k = \FoldPI{a}{c=0}{\vect{a \le \Max; i=i_2 \Then a = \Max; a=\Max \;\Then c\pp}} \fand  \\
% 0 \le i_1 < |a| \fand  0 \le i_2 < |a| \fand j=k
%\end{gather*}
%\begin{gather*}
%\bm{\vect{|a|;j}} \!=\! \Fold\!\bm{\vect{\bar{a}=0; c=0}}\!\!\bm{\vect{ a \ge min; \bar{a}=m \;\Rightarrow\; a = min; 
%a=min \;\Rightarrow\; c++}}  \;\land\;
%\bm{\vect{|a|;k}} \!=\! \Fold\!\bm{\vect{\bar{a}=0; c=0}}\!\!\bm{\vect{ a \le max; \bar{a}=M \;\Rightarrow\; a=max;  
%a=max \;\Rightarrow\; c++}} \;\land \\
% 0 \le m < |a| \; \land \;  0 \le M < |a| \;\land \; j=k
%\end{gather*}
%\begin{gather*}
%\bm{\vect{|a|;j}} \!\!=\!\! \Fold\!\bm{\vect{a=0}}\bm{\vect{i++; i=m \land a \ne min \Rightarrow \Break; a < min \Rightarrow \Break; 
%a=min \;\Rightarrow\; c++}}\!\!\bm{\vect{i=0;c=0}}  \;\land\;
%\bm{\vect{|a|;k}} \!\!=\!\! \Fold\!\bm{\vect{i++; i=M \land a \ne max \Rightarrow \Break; a > max \Rightarrow \Break; 
%a=max \;\Rightarrow\; c++}}\!\!\bm{\vect{i=0;c=0}}\!\!(\vect{a}) \\
%\land \; j=k
%\end{gather*}
\end{small}
\vspace*{-2mm}
\caption{Array Folds Logic}
\label{fig:toy-c}
\end{subfigure}
\caption{A toy array problem}
\label{fig:toy}
%\vspace*{-4mm}
\end{figure}

%or, alternatively, 
%$ \bm{\vect{|a|;m; m}} \!=\! \Fold\!\bm{\vect{i++; a < min \Rightarrow \Break; 
%a=min \;\Rightarrow\; j++; a > max \Rightarrow \Break; 
%a=max \;\Rightarrow\; k++}}\!\!\!\bm{\vect{i=0;j=0;k=0}}\!\!(a)  \;\land\; m>0
%$
%
%$\Fold\!\bm{\vect{0}}\!\!\bm{\vect{\;\;\top \;\;\Rightarrow\;a \ge min; i=k \;\Rightarrow\; a=min; i>k \;\Rightarrow\; a \ne min}}\!\!(a)  \land 
%\Fold\!\bm{\vect{0}}\!\!\bm{\vect{\;\;\top \;\;\Rightarrow\;a \le max; i=l \;\Rightarrow\; a=max; i<l \;\Rightarrow\; a \ne max}}\!\!(a) $
%
%$\bm{\vect{|a|}} =   \Fold\!\bm{\vect{0}}\!\!\bm{\vect{\;k\le i\le l \;\Rightarrow\;c=c+1}}\!\!(a) $
%
%\ \\
%
%$\exists i_{min},min, max . (a[k] = min \land \forall i. a[i] \ge min \land \forall i>k. a[i] \ne min) \; \land$
%
%$\exists k,min . (a[k] = min \land \forall i. a[i] \ge min \land \forall i>k. a[i] \ne min)$
%
%
%$ 1 2 5 3 1 5 4 2 1 $

The solution we propose is shown in Figure~\ref{fig:toy-c}:
%. It looks very much like the programming notation above, but it is actually a logic of arrays with \Fold\ functions. 
in the example formula, the first \Fold\ applies a function to array $a$. The vector in the first parentheses gives initial values for the array index and counter $\cnt_1$; the function is folded over the array starting from the initial index. Index variable $\idx$ is implicit, and it is incremented at each iteration. The function itself is given in the second parentheses, and has two branches. The first branch \emph{counts} the number of positions with elements equal to $\Min$ in counter $\cnt_1$. The second branch \Skip{s} when the current array element $\elm$ is greater than the (guessed, existentially quantified) variable $\Min$. When $\elm<\Min$, the implicit \Break\ statement is executed, and the \Fold\ terminates prematurely.
% Note that such unconditional branches in \Fold\ terms can represent assertions that would otherwise require universal quantification. 
The result of the \Fold\ is compared to the vector which asserts that the final value of the array index equals to the array size $|a|$ (which means no \Break\ was executed), and the final value of $\cnt_1$ equals to $j$. The positions where elements are equal to $\Max$, are counted in the second \Fold, and the equality between these two counts is asserted. The ability to count over arrays with unbounded elements is a unique feature of Array Folds Logic.
%three branches. The first branch is executed unconditionally, and checks that every array element is greater or equal to the (guessed, existentially quantified) variable $\Min$. The second branch checks that the value at the guessed index $i_1$ is indeed equal to $\Min$. The last branch \emph{counts} the number of positions with elements equal to $\Min$ in counter $c$, which is initialized to $0$.

This paper makes the following contributions:

{\textbf{1.}} We define a new logic, called AFL, that can express interesting and non-trivial properties of counting over arrays, which are orthogonal to the properties expressible by other logics. Additionally, AFL can concisely summarize loops with internal branching that traverse arrays and perform counting, enabling verification and symbolic testing of programs with such loops.

{\textbf{2.}} We show that the satisfiability problem for AFL is $\PSPACE$-complete, and with a natural restriction the complexity decreases to $\NP$. We provide a decision procedure for AFL, which works by a reduction to the emptiness of (symbolic) reversal-bounded counter machines, which in turn reduces to the satisfiability of existential Presburger formulas. We show that adding either universal quantifiers or concatenation leads to undecidability.

{\textbf{3.}} We implemented tool \textsc{AFolder} \cite{tool} that can discharge proof obligations in AFL, and we demonstrate on real-life examples that our decision procedure can solve a broad range of problems in symbolic testing and program verification. %The tools is available at \cite{tool}.

\paragraph{Related work.} Our logic is related to the quantified fragments of the theory of arrays such as \cite{Bradley2006}\cite{HabermehlIV08}\cite{Habermehl2008b}\cite{AlbertiGS15}. These logics allow restricted quantifier prefixes, and their decision procedures work by rewriting to the (parametric) theories of array indices and elements (Presburger arithmetic being the most common case) \cite{Bradley2006}\cite{AlbertiGS15}, or by reduction to flat counter automata with difference bound constraints \cite{HabermehlIV08}\cite{Habermehl2008b}. An interesting alternative is provided in \cite{zhou2014array}, where the quantification is arbitrary, but array elements must be bounded by a constant given a priori; the decision procedure works by a reduction to WS1S.  A separate line of work is presented by the theory combination frameworks of \cite{Ghilardi2007}\cite{DeMoura09}, where the quantifier-free theory of arrays is extended by  \var{injective} predicate and  \var{domain} function \cite{Ghilardi2007}, or with \var{map} and \var{constant}-\var{value} combinators \cite{DeMoura09}. The theory of concatenation and its extensions \cite{Buchi88}\cite{Furia10}\cite{LinB16} are also related; their decision procedures work by reduction to Makanin's algorithm for solving word equations \cite{makanin1977}.
AFL can express some properties that are also expressible in these logics, such as \emph{boundedness}, \emph{partitioning}, or \emph{periodicity}; other properties, such as \emph{sortedness}, are not expressible in AFL. The counting properties that constitute the core of AFL are not expressible in any of the above logics.
We compare the expressive power of AFL and other logics in Section~\ref{sec:express}.

%There are two cornerstone operations, random access and concatenation, that form the basis of array and list theories. Correspondingly, different logics can be broadly divided into those that allow one or the other. Array logics can allow some restricted form of quantifier alternation, and the decision procedure typically works by reduction to Presburger arithmetic \cite{Bradley2006,.,.} or automata; logics of lists with concatenation are restricted to the quantifier-free fragment, and are usually based on Makanin's algorithm for solving word equations. 

%Our logic is related to both directions outlined above. In particular, it allows random access and a restricted form of list operations. There is a substantial set of properties that are expressible both in AFL and in these logics. At the same time, AFL allows to express some important properties, that are not expressible in other logics.
%
There are numerous works on loop acceleration and summarization \cite{Comon1998}\cite{Bozga2010}\cite{KnoopKZ11}, also in the context of verification and symbolic testing \cite{SaxenaPMS09}\cite{Flata2009}\cite{GodefroidL11}\cite{Hojjat2012} and array-manipulating programs \cite{Bozga2009}\cite{Booster2014}\cite{Alberti2015}. Our logic allows one to summarize loops with internal branching and counting, which are outside of the scope of these works.

The decision procedure for AFL is based on decidability results for emptiness of reversal-bounded counter machines \cite{Ibarra78}\cite{Ibarra81}\cite{Ibarra02}, on the encoding of this problem into Presburger arithmetic \cite{HagueL11}, and on the computation of Parikh images for NFAs \cite{seidl2004counting}. In Section~\ref{sec:procedure} we extend the encoding procedure to symbolic counter machines, and present some substantial improvements that make it efficient for solving practical AFL problems.

%%% Local Variables:
%%% mode: latex
%%% TeX-master: "main"
%%% End:

%\section{Preliminaries}
%\label{sec:preliminaries}
%
%\input{preliminaries}

\section{Array Folds Logic}
\label{sec:logic}
 We assume familiarity with the standard syntax and terminology of many-sorted first-order logics. We use vector notation: $\vec{v} = (v_1,\ldots,v_n)$ denotes an ordered sequence of terms. For two vectors $\vec{u}$ and $\vec{v}$, we write their concatenation as $\vec{u}\vec{v}$.

%\subsection{Notation}
%
%For notational convenience, we use vector notation:
%$\vec{w} = (w_1,\ldots,w_n)$ to denote an ordered sequence of terms.
%Given two vectors $\vec{w}$ and $\vec{v}$, we denote their
%concatenation as $\vec{w}\vec{v}$.
%
%A formula $\vec{x}=\vec{y}$ means that the sizes of vectors $\vec{x}$
%and $\vec{y}$, as well as the corresponding elements $x_i$ and $y_i$,
%are equal.
%
%Everywhere in the paper, if not stated otherwise, we use the following
%convention: $a$, $b$ are array variables; $i$, $j$ are integer
%variables; $m$, $n$ are integer constants.

%\paragraph{Theory of Arrays.}

%A set $\mathcal{Q}$ of strings over $\{\exists,\forall\}$ denotes the $Q$-fragment of a first-order theory: all formulas in the fragment are of the form $\mathcal{d} \partial_1 v_1 \cdots \partial_n \varphi$, where $\varphi$ is a quantifier-free formula.

Within this paper we consider the domains of arrays, array
indices, and array elements to be $\Ab = \Zb^*$,
$\Nb = \{\,0,1,\ldots \,\}$, and $\Zb = \{\,\ldots,-1,0,1,\ldots \,\}$
respectively.

Presburger arithmetic has the signature
$\Sigma_{\mathbb{Z}} = \{\,0,1,+,<\,\}$; we use it for array indices
and elements, as well as other arithmetic assertions, possibly with
embedded array terms.
%$\Sigma_{\mathbb{Z}} = \{\,0,1,+,-,=,<\,\}$
We write $\true$ and $\false$ to denote a valid and an unsatisfiable Presburger formula, respectively.

The theory of integer-indexed arrays extends Presburger arithmetic
with functions \var{read}, and \var{write}, and has the signature
$\Sigma_{A} = \Sigma_{\mathbb{Z}} \cup \{\,\cdot[\cdot],\cdot\{\cdot
\leftarrow \cdot\}\, \}$.
The \var{read} function $a[i]$ returns the $i$-th element of array
$a$, and the \var{write} function $a\{i \leftarrow x\}$ returns array $a$ where the $i$-th element is replaced by $x$. These functions should satisfy the \emph{read}-over-\emph{write} axioms as described by McCarthy \cite{McCarthy62}; the decision procedure for the quantifier-free array theory is presented in \cite{QuantifierFreeArrays}.

\subsection{Syntax}
\label{afl-syntax}

\emph{Array Folds Logic (AFL)} extends the quantifier-free theory of integer arrays with the ability to perform counting. The extension works by incorporating \Fold\ terms into arithmetic expressions; such a term folds some function over the array by applying it to each array element consecutively.

AFL contains the following sorts: array sort $\asort$, integer sort $\isort$, Boolean sort $\bsort$, and two enumerable sets of sorts for integer vectors $\vsort{m}$ and functional constants $\fsort{m} = \vsort{m} \times \isort \rightarrow \vsort{m}$, for each $m \in \Nb$, $m>0$. The syntax of the AFL terms is shown in Table~\ref{tab:syntax}; $a$ and $b$ denote array variables, $x$ denotes an integer variable, $n$ and $m$ denote integer constants.

\emph{Array terms}~$A$ of sort $\asort$ are represented either by an array variable $a$, or by the \var{write} term $a\{T \leftarrow T\}$.

\emph{Integer terms}~$T$ of sort $\isort$ can be integer constants $n \in \Zb$, integer variables $x$, integer addition, \var{read} term $a[T]$ for the index represented as an integer term, or the term $|a|$, which represents the length of array $a$. 

\emph{Boolean terms} $B$  of sort $\bsort$ are formed by array equality, usual Presburger and Boolean operators, and equality between vectors of sort $\vsort{m}$.

\emph{Vector terms} $V^m$ of sort $\vsort{m}$ are either a list of $m$ integer terms, or a \Fold\ term. The former is written as a vertical list in parentheses; they can be omitted when $m=1$. The latter, written as $\Fold_{a}\:{v}\:{f}$, represents the result of the transformation of an input vector $v$ of sort $\vsort{m}$ by folding a functional constant $f$ of sort $\fsort{m}$ over an array $a$.  The first element of $v$ specifies an initial value of the array index; the remaining elements give initial values for the counters that can be used inside $f$. The resulting vector after the transformation gives the final values for the array index and the counters.

\emph{Functional constants} (when no confusion can arise, we call them  \emph{functions}) $F^m$ of sort $\fsort{m}$ can only be  a parenthesized list of branches (guarded commands); the length of the list is unrelated to $m$. %, which defines the number of counters a function operates on. 
A function $f$ of sort $\fsort{m}$ can refer to the following implicitly declared variables: $\elm$ for the currently inspected array element; $\idx$ for the current array index; $\cnt_1,\ldots,\cnt_{m-1}$ for the counters; $\st$ for the state (control flow) variable. All other variables that occur inside $f$ are considered as free variables of sort $\isort$. 

\emph{Guards} are conjunctions of \emph{atomic guards}, which can compare array elements, indices, and counters to integer terms; the state variable can only be compared  to integer constants. 
\emph{Updates} are lists of \emph{atomic updates}; they can increment or decrease counters by a constant, assign a constant to the state variable, \Skip, i.e. perform no updates, or execute a \Break\ statement, which terminates the \Fold\ at the current position.  Counter or state updates define a function $\Zb \ra \Zb$.
Guards and updates translate into logical formulas that either constraint the current variable values, or relate the current and the next-state (primed) variable values in the obvious way; we denote this translation by $\Phi$. E.g.,
% the guard $\guard \equiv (\cnt_1 > T)$ defines the formula $\Phi(\guard) \equiv (\cnt_1 > T)$, while 
 the update $\update \equiv (\cnt_1 \pen)$ defines the formula $\Phi(\update) \equiv (\cnt_1' = \cnt_1+n)$. 

%The \Skip\ statement with the preceding $\Rightarrow$ can be omitted.
\begin{table}[t]
\begin{align*}
  A~~~::=~& a ~~|~~ a\{T \leftarrow T\} \\
% ~~|~~ \Filter_{a}(f)  ~~|~~ \Map_{a}(f) ~~|~~ \var{tail}(a)\\
%  \Filter_{a}(\lambda\, x\,.\,F(x))  ~~|~~ \Map_{a}(\lambda\, x\,.\,T(x))
  T~~~::=~& n ~~|~~ x ~~|~~ T+T ~~|~~ a[T] ~~|~~ |a| \\
% ~~|~~ \var{head}(a) \\
  B~~~::=~& a = b  ~~|~~ T = T ~~|~~ T < T ~~|~~ \neg B ~~|~~ B \land B ~~|~~    
  V^m = V^m \\
  V^m~::=~& \Vect{T;\cdots;T} ~~|~~ \Fold_{a}\: V^m\: F^m 
  ~~~~~~~~~~~~~~~~~~~~~~~~~~~~ \\
%  F^m~::=~& \Vect{\guard \!\Then\! \update;\cdots;\guard \!\!\Then\!\! \update} \\
  F^m~::=~& \Vect{\guard \!\Then\! \update;\cdots;\guard \!\Then\! \update} \\
%  \branch ~::=~&  \guard \!\!\Then\!\! \update \\
  \guard ~::=~&  \elm \approx T ~|~ \idx \approx T ~|~ \cnt_m \approx T  ~|~  \st \approx  n ~|~ \guard \land \guard \hspace{15mm}\text{($\approx \;\in\: \{>,<,=, \ne\}$)}\\
  \update ~::=~&  \cnt_m \pen ~|~ \st \leftarrow n ~|~ \Skip ~|~ \Break ~|~  \update \:; \update
 % \hspace{16.7mm}\text{($\sim \;\in\: \mset{\texttt{++},\texttt{--} \texttt{+=}\,n,\texttt{-=}\,n}$)}
\end{align*}
\vspace*{2mm}
\caption{Syntax of AFL.}
\label{tab:syntax}
\end{table}

%In order that functions behave in a functional way, 
We require that guards of all branches are mutually exclusive. There is an implicit ``catch-all'' branch with the \Break\ statement, whose guard evaluates to $\true$ exactly when guards of all other  branches evaluate to $\false$. We also require that each branch contains at most one update for each implicit variable.

%For the logic to be decidable, w
We restrict the control flow in functions, which is defined by state variable $\st$. Notice that $\st$ is syntactically finite state. Thus, given a set of function branches $\var{Br}$, we define an edge-labeled control flow graph $G = \tuple{S,E, \gamma}$, where:
%\vspace*{-0mm}
\begin{itemize}
\item states $S = \bigmset{0} \cup \bigmathset{n}{\st \!\leftarrow\! n \in \var{Br}}$;
\item edges $E = \bigcup_{\guard \!\!\Then\!\! \update \:\in\: \var{Br}} \bigmathset{(s_1,s_2)}{s_1 \!\models\! \guard \land s_2 \!=\! \var{ite}(\st \!\leftarrow\! n \in \update,n,s_1)}$;
\item $\gamma$ is the labeling of edges with the set of formulas $\Phi(\guard)$ and $\Phi(\update)$ for each guard or update which occurs in the same branch.
\end{itemize}
%\vspace*{-0mm}
% contains $0$ and the set of constants from $\st \ra n$ terms
%$E$ is defined by $\st \ra n$ terms,
%,  and $\gamma$ is the labeling of edges with $\cnt_n \pen$ terms that occur in the same update as $\st \ra n$ terms. 
We require that edges in the strongly-connected components of $G$ are labeled with counter updates that are, for each counter, all non-decreasing, or all non-increasing. Thus, $G$ is a DAG of SCCs, where counters within each SCC behave in a monotonic way. We use this restriction to derive from $f$ a reversal-bounded counter machine (see Definition~\ref{def:fscm}).

The presented syntax is minimal and can be extended with convenience functions and predicates such as $\{-, n \cdot, , \le, \ge, \for, \texttt{++},\texttt{--},\texttt{-=}n \}$ in the usual way. We allow to use $*$ to denote the absence of constraints: this is useful for vector notation. We replace each $*$ in the formula with a unique unconstrained variable.

%For convenience, we allow also branches of a slightly extended syntax: $\guard \land \update$, and $\guard  \Rightarrow \guard \land \update$; it allows guards to occur in a conjunction with updates. The intended meaning is that such a guard, whenever applicable, should hold; otherwise, the \Fold\ is terminated. The extended syntax translates to the standard one by applying the following rule till saturation: replace a command $g_1 \land g \land u$, where $g_1$ is an atomic guard, with two commands $\neg g_1 \Rightarrow break$, and $g_1 \Rightarrow g \land u$ (note that atomic guards are closed under negation). Each application of this rule reduces by one the number of atomic guards, which occur in a conjunction with updates.

\subsection{Semantics}

\begin{table}[t]
\begin{align*}
   1.~& \ev{\Vect{t^1_1;\cdots;t^1_m}\! =\! \Vect{t^2_1;\cdots;t^2_m}}\!\!\!\!\!\!\!\!\!\!\! &~\equiv~  &\left(\ev{t^1_1} = \ev{t^2_1} \right) \land \ldots \land \left(\ev{t^1_m} = \ev{t^2_m} \right) \\
%2.~& \ev{\Vect{t_1;\cdots;t_m}} &~\equiv~  &\Vect{\ev{t_1};\cdots;\ev{t_m}}\\
 2.~& \ev{\Fold_{a}\: v\: f} &~\equiv~  &\evsk{\Fold_{a}\: v\: f} \text{, where } \kappa(\FV) = \Vect{v;0}\\   
3.~&\evsk{\Fold_{a}\: v\: f} &~\equiv~ &\text{if  } (\evk{\idx}<0) \text{ or } (\evk{\idx}\ge |a|) \text{ or } (\var{false} \in \evsk{f}) \text{ then } v\\  
& & &\text{else} \evskp{\Fold_{a}\: v'\: f}\!\!\!\text{, where } \kappa'(\FV) = \Vect{v';\st'} = \evsk{f}\big(\kappa(\FV)\big)\\  
4.~&\evsk{f}\Vect{v_1;\ldots;v_m} &~\equiv~  & \Vect{v'_1;\ldots;v'_m}, \text{ where } v'_j \equiv \text{ if } \var{upd}(v_j) \in \evsk{f} \text{ then } \var{upd}(v_j) \text{, else } v_j \\
5.~&\evsk{\Vect{\guard_1 \Rightarrow \update_1;\cdots;\guard_m \Rightarrow \update_m}} &~\equiv~  & \{\idx'\!=\!\idx\!+\!1\} \cup \evsk{\guard_1 \Rightarrow \update_1} \cup \ldots \cup  \evsk{\guard_m \Rightarrow \update_m}\\  
6.~&\evsk{\guard \Rightarrow \update} &~\equiv~ &\text{if } \evsk{\guard}=\var{true} \text{  then  } \evsk{\update} \text{  else  } \emptyset\\
7.~&\evsk{ \elm \approx t} &~\equiv~  &\evk{\elm} \approx \ev{t} \hspace{20mm}\text{(similarly for } \idx \approx T ,\; \cnt_m \approx T,\;    \st \approx  n \text{)}\\  
%&\evsk{ \idx \approx t} &~\equiv~  &\evk{\idx} \approx \ev{t}\\  
%&\evsk{ \cnt_m \approx t} &~\equiv~  &\evk{\cnt_j} \approx \ev{t}\\  
%&\evsk{ \st \approx n} &~\equiv~  &\evk{\st} \approx n\\  
8.~&\evsk{\guard_1 \land \guard_2} &~\equiv~  &\evsk{\guard_1} \land \evsk{\guard_2} \\  
9.~&\evsk{\update_1 ; \update_2} &~\equiv~  &\evsk{\update_1} \cup \evsk{\update_2} \\  
10.~&\evsk{\cnt_m \pen } &~\equiv~ & \{\cnt_m' \!=\! \cnt_m \!+\! n \}\\
%\hspace{15mm}\text{(similarly for } \texttt{++}, \texttt{--},\texttt{-=}n \text{)}\\  
11.~&\evsk{\st \leftarrow n } &~\equiv~  & \{\st' = n \}\\  
12.~&\evsk{\Skip} &~\equiv~ & \emptyset\\
13.~&\evsk{\Break} &~\equiv~ & \mset{\var{false}} 
% \left[ \right]^\sigma &~\equiv~  \left[ \right]^\sigma \\  
% \ev{\Vect{b_1;\cdots;b_k}} &~\equiv~  \ev{b_1} \land \ldots \land \ev{b_k} \land \bigwedge_{x' \in V \,.\, \forall j . x' \not\in \var{vars}({\ev{b_1}}) }{x'=x}
\end{align*}
\vspace*{2mm}
\caption{Semantics of AFL}
\label{tab:semantics}
\end{table}

%Semantically, a functional constant $f$ of sort $\fsort{m}$ is a function that operates on two sets of variables: the \var{pre}-variables $\{\idx,c_1,\ldots,c_{m-1},s, b, \elm \}$ and the  \var{post}-variables $\{\idx',c_1',\ldots,c_{m-1}',s', b'\}$. Here $b$ is a special variable for handling \Break.
%
%Guards and updates translate to atomic formulas in the standard way; we denote the result of translating $\delta$ by $\tau(\delta)$. E.g., we have $\tau(c_1 \pe 2) \equiv c_1' = c_1 + 2$, and $\tau(s \leadsto 3) \equiv s' = 3$. For translating \Break, we introduce a dedicated integer variable $b$: it starts initialized with $0$, and $\tau(\Break) \equiv b'=1$.
For a given AFL formula $\phi$, we denote the sets of free variables of $\phi$ of sort $\asort$ and $\isort$ by $\VA$ and $\VI$, respectively. All free variables are implicitly existentially quantified. For functions of sort $\fsort{m}$, we denote by $\FV$ the set of their implicit variables $\{\idx,\cnt_1,\ldots,\cnt_{m-1},\st \}$.
%of $\phi$ of sort $\fsort{m}$ is denoted by $\CF{m}$.

The treatment of array writes and reads in the quantifier-free array theory is standard \cite{QuantifierFreeArrays}, and we do not elaborate on it here.
Array equalities partition the set of array variables into equivalence classes; all other constraints are then translated into constraints over a representative of the corresponding equivalence class.

An \emph{interpretation} for AFL is a tuple $\sigma = \tuple{\lambda,\mu}$, where
%\begin{itemize}
%\item 
$\lambda: \VI \ra \Zb$ assigns each integer variable an integer, and
%\item
 $\mu: \VA \ra \Zb^*$ assigns each array variable a finite sequence of integers.
%\end{itemize}

%\item $\nu: \CF{m} \ra (\Zb^{m+3} \ra \Zb^{m+2})$, for each $m>0$, assigns each functional constant of sort $\fsort{m}$ a function that maps an integer vector of size $m+3$ to an integer vector of size $m+2$.
%\begin{itemize}
%\item $\lambda: \VI \ra \Zb$, assigns each integer variable an integer;
%\item $\mu: \VA \ra \Zb^*$, assigns each array variable a finite sequence of integers.
%%\item $\nu: \CF{m} \ra (\Zb^{m+3} \ra \Zb^{m+2})$, for each $m>0$, assigns each functional constant of sort $\fsort{m}$ a function that maps an integer vector of size $m+3$ to an integer vector of size $m+2$.
%\end{itemize}
The semantics of an AFL term $t$ under the given interpretation $\sigma$ is defined by the \emph{evaluation} $\left[ t \right]^\sigma$. 
Terms that constitute functions are evaluated in the additional \emph{context} $\kappa$. For a function $f$ of sort $\fsort{m}$, $\kappa: \FV \ra \Zb^{m+1}$ maps internal variables of $f$ to integers. 
%In order that $f$ behaves in a functional way, we require that for each pair of commands $g_1\!\!\Then\!\! u_1$  and $g_2 \!\!\Then\!\! u_2$ in $f$, either
%%one of the following holds:
%%\begin{enumerate}
%%\item 
%the formula $\evsk{g_1 \fand g_2} \equiv \var{false}$ for all $\sigma$ and $\kappa$; or
%%\item 
%a) the sets of variables from $\FV$ in $u_1$ and $u_2$ are disjoint, and b) either both $u_1$ and $u_2$ contain a $\Break$, or none of them does. 
%%\end{enumerate}
The evaluation of Presburger, Boolean, and array terms is standard; the remaining ones are shown in Table~\ref{tab:semantics}. We give some explanations here (the remaining semantic rules are self-explanatory):
\begin{enumerate}
\item Vector equality resolves to a conjunction of equalities between components.
%\item An evaluation of a vector is a vector of component evaluations.
\item A \Fold\ term evaluates in the initial context that is defined by the given initial vector of counters $v$, and assigns 0 to state variable $\st$.
\item A contextual \Fold\ term checks whether the array index is out of bounds, or a $\Break$ statement is executed in the current context (this is the only way for $\evsk{f}$ to contain $\false$). If yes, \Fold\ terminates, and returns the current vector $v$. Otherwise \Fold\ continues with the updated vector and context.
\item If an update $\var{upd}(v_j)$ for some variable $v_j$ is present in the function evaluation, then it is applied. Otherwise, the old variable value is preserved.
\item An evaluation of a function, represented by a list of branches, is a union of updates from its branch evaluations. Index $\idx$ is always incremented by 1.
\item A guarded command evaluates to its update if its guard evaluates to \var{true}.
\item A comparison over an internal variable evaluates it in the context $\kappa$, and the comparison term is evaluated in the interpretation $\sigma$.
%\item Both components of a conjunction are evaluated separately.
%\item Updates are evaluated, and collected in a set of updates for the branch.
\end{enumerate}

\subsection{Expressive Power}
\label{sec:express}

Here we give some example properties that are expressible in AFL, and compare its expressive power to other decidable array logics.

\begin{enumerate}

\item \textbf{Boundedness}. All elements of array $a$ belong to the interval $[l,u]$.
%\vspace*{-1mm}
\begin{gather*}
\FoldPI{a}{0}{\vect{l \le \elm \le u \Then \Skip}} = |a|
\end{gather*}

\item \textbf{Partitioning}. Array $a$ is partitioned if there is a position $p$ such that all elements before $p$ are smaller or equal than all elements at or after $p$.
%\vspace*{-1mm}
\begin{gather*}
\FoldPI{a}{0}{\vect{\idx<p  \fand \elm  \le a[p]  \Then \Skip; 	\idx \ge p \fand  \elm \ge a[p]  \Then \Skip}} = |a|
%\begin{gather*}
%\FoldPI{a}{0}{\vect{\idx<p  \Then \elm  \le a[p]; 	\idx \ge p \Then  \elm \ge a[p] }} = |a|
\end{gather*}

\item \textbf{Periodicity}. Array $a$ is of the form $(01)^*$:
%\vspace*{-1mm}
\begin{gather*}
 \FoldPI{a}{0}{\vect{\st  = 0 \fand \elm=0  \Then \st \leftarrow 1; \st = 1 \fand \elm=1  \Then \st  \leftarrow 0}} = |a|
\end{gather*}
% \FoldPI{a}{0}{\vect{\st  = 0 \Then \elm=0 \:\land\: \st \ra 1; \st = 1 \Then \elm=1 \:\land\: \st  \ra 0}} = |a|

\item \textbf{Pumping}. Array $a$ is of the form $0^n 1^n$ (a canonical non-regular language; $0^n 1^n 2^n$, a non-context-free language, is equally expressible):
%\vspace*{-1mm}
\begin{gather*}
 \FoldPI{a}{\vect{0;0;0}}{\vect{\st = 0 \fand \elm=0  \Then \cnt_1\pp \phantom{ \st \leftarrow 1 \fand}; \st = 0 \fand \elm=1  \Then \cnt_2\pp \fand \st \leftarrow 1; \st = 1 \fand \elm=1 \Then \cnt_2\pp \phantom{ \st \leftarrow 1 \fand}}} = \Vect{|a|;n;n}
\end{gather*}

\item \textbf{Equal Count}. Arrays $a$ and $b$ have equal number of elements greater than $l$:
%\vspace*{-1mm}
\begin{gather*}
 \Vect{|a|;n} = \FoldPI{a}{\vect{0;0}}{\vect{\elm > l \Then \cnt_1\pp; \elm \le l \Then \Skip}} \fand \Vect{|b|;n} = \FoldPI{b}{\vect{0;0}}{\vect{\elm > l \Then \cnt_1\pp; \elm \le l \Then \Skip}}
\end{gather*}

\item \textbf{Histogram}. The histogram of the input data in array $a$ satisfies the distribution $H\big(\{ i \:|\: a[i]<10 \}\big) \ge  2 H\big(\{ i \:|\: a[i] \ge 10 \}\big) $:
%\vspace*{-1mm}
\begin{gather*}
 \FoldPI{a}{\vect{0;0}}{\vect{\elm<10 \Then \cnt_1\pp; \elm\ge 10 \Then \Skip}} = \Vect{|a|;h_1} \fand
 \FoldPI{a}{\vect{0;0}}{\vect{\elm \ge 10 \Then \cnt_1\pp; \elm < 10 \Then \Skip}} = \Vect{|a|;h_2} \fand
 h_1 \ge 2 h_2 
\end{gather*}

\item \textbf{Length of Format Fields}. The array contains two variable-length fields. The first two elements of the array define the length of each field; they are followed by the fields themselves, separated by $0$:
%\vspace*{-1mm}
\begin{gather*}
 \mathit{len}_1=a[0] \fand \mathit{len}_2=a[1] \fand 
 \FoldPI{a}
        {\vect{2;0;0}}
        {\vect{\st = 0 \fand \elm\ne 0 \Then \cnt_1\pp; \st = 0 \fand \elm=0 \Then \st \leftarrow 1; \st = 1 \fand \elm \ne 0 \Then  \cnt_2\pp}} = \Vect{|a|;\var{len}_1;\var{len}_2}
\end{gather*}

\end{enumerate}

\paragraph{Comparison with other logics.} Most decidable array logics can specify universal properties over a \emph{single} index variable like (1) above; AFL uses \emph{fold}s to express such universal quantification. Properties that require universal quantification over \emph{several} index variables, like sortedness, are inexpressible in AFL (it can simulate some of such properties, like partitioning (2), using a combination of \emph{fold}s with existential guessing).
% are expressible in most array logics, including AFL, with a \emph{variable} position $p$, but in the logic of sequences \cite{Furia10} it is expressible only for a \emph{fixed} position.
Periodic facts like (3) are inexpressible in \cite{Bradley2006}, but AFL as well as \cite{HabermehlIV08}\cite{Furia10} can express it.  Counting properties such as (4)--(7), which constitute the core of AFL, are not expressible in other decidable logics over arrays and sequences. % other properties, such as \emph{sortedness}, are not expressible in AFL. The counting properties that constitute the core of AFL are not expressible in any of the above logics.

\section{Motivating Example}
\label{sec:example}
As a motivating example to illustrate applications of our logic, we consider a parser for the Markdown language as implemented in the Redcarpet project, hosted on GitHub \cite{Markdown}. Redcarpet is a popular implementation of the language, used by many other projects, in particular by the GitHub itself. Figure~\ref{fig:example} shows the excerpt  from the function \texttt{parse\_table\_header}, which can be found in the file \texttt{markdown.c}. %For the clarity of the presentation we show only the essential parts of the code. %; due to space constraints we also had to squeeze several lines of code into one. Written properly, these program fragment has the size of 32 LOC: a serious challenge for testing and verification, as we show further.

The function considered in the example parses the header of a table in the Markdown format. The first line of the header specifies column titles; they are separated by pipe symbols (`\texttt{|}'); the first pipe is optional. Thus, the number of pipes defines the number of columns in the table. The second line describes the alignment for each column, and should contain the same number of columns; in between each pair of pipes there should be at least three dash (`\texttt{-}') or colon (`\texttt{:}') symbols. A colon on the left or on the right side of the dashes defines left or right alignment; colons on both sides mean centered text. Thus, the two lines ``\texttt{|One|Two|Three|}'' and ``\texttt{|:--|:--:|--:|}'' specify three columns which are left-, center-, and right-aligned. Replacing the second line with either ``\texttt{|:-|:--:|--:|}'' or ``\texttt{|:--|:--:|}'' would result in the ill-formed input: the former doesn't contain enough dashes in the first column, while the latter doesn't specify the format for the last column.

\newcommand{\annotate}[1]{
\vspace*{-1mm}
$\Big\{\;
#1
\:\Big\}$
\vspace*{-0.2mm}
}

\begin{figure}
\begin{flushright}

\begin{small}

\begin{verbatim}
 1: static size_t parse_table_header(uint8_t *a, size_t size, ...)
 2:   size_t i=0, pipes=0;
\end{verbatim}
\hspace{2mm}\annotate{
i_0 = 0 \fand p_0=0
}
\begin{verbatim}
 3:   while (i < size && a[i] != '\n')
 4:     if (a[i++] == '|') pipes++;
\end{verbatim}
\annotate{
\Vect{i_1;p_1} = \FoldPI{a}{\vect{i_0;p_0}}{\vect{\elm=P \Then \cnt_1\pp; \elm \ne P \fand \elm \ne N \Then \Skip}}
%p_1 = \FoldPI{a[i_0,i_1]}{p=p_0}{\vect{a \ne N; a=P \Then p\pp}}
}
% if (i == size || pipes == 0) return 0;
%   header_end = i;
%  while (header_end > 0 && _isspace(data[header_end - 1]))
%    header_end--;
\begin{verbatim}
 5:   if (a[0] == '|') pipes--;
\end{verbatim}
\annotate{
%(a[0]=P \fand p_2=p_1-1) \for (a[0] \ne P \fand p_2 = p_1)
\Vect{*;p_2} = \FoldPI{a}{\vect{0;p_1}}{\idx=0 \fand \elm=P \!\!\Then\!\! \cnt_1\mm}
%p_2 = \FoldPI{a[0,*]}{p=p_1}{i=0 \fand a=P \fand p\mm}
}
%  if (header_end && data[header_end - 1] == '|')
%    pipes--;
%  *columns = pipes + 1;
%  *column_data = calloc(*columns, sizeof(int));
% i++
%  if (i < size && data[i] == '|')
%    i++;
\begin{verbatim}
 6:  i++;
 7:  if (i < size && a[i] == '|') i++;
\end{verbatim}
\annotate{
i_2=i_1+1 \fand i_3 = \FoldPI{a}{i_2}{\idx=i_2 \fand \elm = P \Then \Skip}
%i_2=i_1+1 \fand \FoldP{a[i_2,i_3]}{i=i_2 \fand a = P}
}
\begin{verbatim}
 8:   end = i;
 9:   while (end < size && a[end] != '\n') end++;
\end{verbatim} % TODO add i++
\annotate{
\var{e}_0=i_3 \fand e1 = \FoldPI{a}{e_0}{\elm \ne N \Then \Skip}
%\var{e}_0=i_1 \fand \FoldP{a[\var{e}_0,\var{e}_1]}{a \ne N}
}
\begin{verbatim}
10:  for (col = 0; col<pipes && i<end; ++col) {
11:    size_t dashes = 0;
\end{verbatim}
\annotate{
c_0=0 \fand c_0<p_2 \fand i_3 < \var{e}_1 \fand d_0=0
}
%    while (i < under_end && data[i] == ' ')
%      i++;
\begin{verbatim}
12:    if (a[i] == ':') { i++; dashes++; column_data[col] |= ALIGN_L; }
\end{verbatim}
\annotate{
%(a[i_1]=C \fand i_2=i_1+1 \fand \var{d}_1=d_0+1) \lor 
%(a[i_1]\ne C \fand i_2=i_1 \fand \var{d}_1=d_0)
\Vect{i_4;d_1} = \FoldPI{a}{\vect{i_3;d_0}}{\idx=i_3 \fand \elm = C \!\!\Then\!\! \cnt_1\pp}
%d_1 = \FoldPI{a[i_3,i_4]}{d=d_0}{i=i_3 \fand a = C \fand d\pp}
}
\begin{verbatim}
13:    while (i < end && a[i] == '-') { i++; dashes++; }
\end{verbatim}
\annotate{
\Vect{i_5;d_2} = \FoldPI{a}{\vect{i_4;d_1}}{\idx<e_1 \fand \elm=D \!\!\Then\!\! \cnt_1\pp} %\var{d}_2 = \FoldPI{a[i_4,i_5]}{d=\var{d}_1}{i<\var{e}_1 \fand a=D \fand d\pp}
}
\begin{verbatim}
14:    if (a[i] == ':') { i++; dashes++; column_data[col] |= ALIGN_R; }
\end{verbatim} % if (i < end && a[i] == ':') {
\annotate{
%(a[i_3]=C \fand i_4=i_3+1 \fand \var{d}_3=d_2+1) \lor 
%(a[i_3]\ne C \fand i_4=i_3 \fand \var{d}_3=d_2)
\Vect{i_6;d_3} = \FoldPI{a}{\vect{i_5;d_2}}{\idx=i_5 \fand \elm = C \!\!\Then\!\! \cnt_1\pp}
%d_3 = \FoldPI{a[i_5,i_6]}{d=d_2}{i=i_5 \fand a = C \fand d\pp}
}
%    while (i < under_end && data[i] == ' ')
%      i++;
\begin{verbatim}
15:    if (i < end && a[i] != '|' && a[i] != '+') break;
16:    if (dashes < 3) break;
17:    i++;
\end{verbatim}
\annotate{
(i_6 \ge \var{end}_1 \lor a[i_6]=P \lor a[i_6]=A) \fand d_3 \ge 3 \fand i_7=i_6+1 
 \fand c_1=c_0+1}
%\annotate{
%i_4 \ge \var{end}_1 \for a[i_4]=P \for a[i_4]=A
%%i_4 \ge \var{end}_1 \for \FoldP{a[i_4,i_4+1]}{a = P} \for \FoldP{a[i_4,i_4+1]}{a = A}
%}
\begin{verbatim}
18:  }
19:  if (col < pipes) return 0;
\end{verbatim}
\annotate{
c_1 \ge  p_2
}
\end{small}
\end{flushright}
\caption{An excerpt from the Redcarpet Markdown parser with AFL annotations}
\label{fig:example}
\vspace*{-3mm}
\end{figure}

Suppose, we are interested in the symbolic testing of the parser implementation; in particular, we want to cover all branches in the code for a reasonably long input. 
%For that we postulate that the first line of the input should contain at least $n$ pipes (we add the condition \texttt{assert(pipes>=n)} after line 5). 
For that we postulate that the first input line contains at least $n$ columns (we add the condition \texttt{assert(col>=n)} after line 19). 

%fix the first line to ``\texttt{| One | Two | Three |}''. 
Now, consider the last conditional statement at line 19. The \texttt{if} branch is satisfied by an empty second input line; and indeed, such concolic testers as \textsc{Crest} can easily cover it. The \texttt{else} branch, however, poses serious problems. In order to cover it, a well-formed input that respects all constraints should be generated; in particular the smallest length of such input, e.g., for $n$ equal to 3, is 17. The huge number of combinations to test exceeds the capabilities of the otherwise very efficient concolic tester: for $n=2$ \textsc{Crest} needs 800 seconds to generate a test, and for $n=3$ it is not able to finish within 3 hours. %Needless to say that no static analysis approach is able to handle this example, when even concolic testing, usually much more scalable than static analysis, fails.

Let us now examine the encoding of the implementation semantics in Array Folds Logic. The AFL assertions are shown in Figure~\ref{fig:example} intertwined with the source code: they encode the semantics of the preceding code lines in the SSA form. To shorten the presentation we use the following conventions: variables \texttt{i}, \texttt{a}, \texttt{pipes}, \texttt{end}, \texttt{col}, and \texttt{dashes} are represented by (SSA-indexed) logical variables $i$, $a$, $p$, $e$, $c$, and $d$ respectively; characters `\texttt{\char`\\ n}',  `\texttt{|}', `\texttt{:}', `\texttt{-}', and `\texttt{+}' by logical constants $N$, $P$, $C$, $D$, and $A$ respectively; finally, the subscript denotes the SSA index of a variable. 

The Presburger constraints such as those after line 2 are standard and we do not elaborate on them here. The first AFL-specific annotation goes after line 4: it directly reflects the loop semantics. The \Fold\ term encodes the computation of the number of pipes: they are computed in the counter $\cnt_1$, which gets its initial value equal to $p_0$, and its final value is equal to $p_1$. Similarly, array index $\idx$ is initialized with $i_0$; and its final value is asserted to be equal to $i_1$. 
 Both for counter $\cnt_1$ and for index $\idx$
 (which is a special type of a counter) 
 their initial and final values can be both constant and symbolic: in fact, arbitrary Presburger terms are allowed. %The function that is folded over an array has two branches. The first one defines a universal condition (the examined symbol is not a newline): if this condition is not met, the \Fold\ is stopped at this position. The second branch is conditional: if the examined symbol is a pipe, then increment counter $\cnt_1$.

Notice that the loop at lines 3-4 is outside of the class of loops that can be accelerated by previous approaches.
%, e.g., by the approaches\citez{Godefroid, DIFF, Sharygina}
In particular, the difficulty here is the combination of the iteration over arrays with the branching structure inside the loop. On the contrary, AFL can summarize the loop in a concise logical formula.

The next conditional statement at line 5, takes care of the optional pipe at the beginning of the input. The annotation shown demonstrates that conditional statements are also easily represented by \Fold\ terms. In particular, here the function is folded over $a$ starting from 0; the final index is unconstrained. The branch checks that the index is 0 (to prevent going further over the array), and that the symbol at this position is `\texttt{|}'. Counter $\cnt_1$ is decremented only if these two conditions are met; otherwise, the \Fold\ terminates. An equivalent encoding using only array reads is possible: $(a[0]=P \fand p_2=p_1-1) \for (a[0] \ne P \fand p_2 = p_1)$, but this encoding involves a disjunction.

The other program statements of the motivating example are encoded in a similar fashion.
%. The encoding is straightforward, and 
%can be easily automated. 
The encoding shown is for one unfolding of the \texttt{for} loop at line 10; several unfoldings are encoded similarly. We have checked the resulting proof obligations 
%for the required number of unfoldings 
with our solver for AFL formulas, called \textsc{AFolder}; 
%it was able to discharge them and generate the required test input in a fraction of a second.
it can discharge them and generate the required test input in less than 2 minutes for $n=3$.% much faster than \textsc{Crest}.

\section{Complexity}
\label{sec:complex}
%\subsection{Counter Machines}
%\label{subsec:cm}

A counter machine is a finite automaton extended by a vector
$\ctrv= (\ctr_1, \ldots, \ctr_k)$ of $k$ counters. Every counter in
$\ctrv$ stores a non-negative integer, and a counter machine can compare
it to constant, and increment/decrease its value by a constant.
% If, due to a decrement, the value of a counter drops below zero, then
%the behavior of the counter machine is undefined. 
For the formal definition of counter machines consult, e.g., \cite{Ibarra02}.

We extend counter machines to symbolic counter machines (SCMs), which accept sequences (arrays) of integers.  We denote the symbolic value of an array cell by a special integer variable $\invar$. Let $X$ be a set of integer variables, where $\invar\not\in X$. An \emph{atomic input constraint} is of the form $\invar \approx c$ or $\invar \approx x$, where $c\in \Nb$, $x\in X$, and $\approx \; \in\{<,\le,>,\ge,=, \ne\}$.   Similarly, an \emph{atomic counter constraint} is a formula of the form $\ctr_i \approx c$ or $\ctr_i \approx x$. An \emph{input constraint} (resp.\ a \emph{counter constraint})
is either a conjunction of $n\geq 1$ atomic input constraints (resp.\ atomic counter constraints), or the formula $\true$.
We denote by $\ic(X)$ (resp. $\cc_k(X)$) the set of all input constraints (resp.\ counter constraints with counters not greater than $k$) over variables in $X$. 

%\begin{definition}
%\label{def:cm}
%  A $k$-\emph{counter machine} is a tuple
%  $\scm = (\ctrv, \alphabet, \states, \tr, \init)$, where
%%\begin{itemize}
%%\item
% $\ctrv= (\ctr_1, \ldots, \ctr_k)$ is a vector of $k$ counter variables,
%%\item 
%$\alphabet$ is a finite set alphabet symbols,
%%\item 
%$\states$ is a finite set of states,
%%\item 
%$\tr \subseteq \states\times \cc_k(\emptyset)\times \alphabet\times\states\times\Zb^k$ is a transition relation,
%%\item
%and $\init \in \states$ is the initial state. A transition $(q_1,c,\alpha,q_2,\vec{\kappa}) \in \delta$ moves the counter machine from state $q_1$ to $q_2$ if the counters satisfy the constraint $c$ and the input symbol $\alpha$ is read; the counters are incremented by $\vec{\kappa}$.
%%\end{itemize}
%\end{definition}

\begin{definition}
\label{def:scm}
  A \emph{symbolic $k$-counter machine} is a tuple
  $\scm = (\ctrv, \vars, \states, \tr, \init)$, where:
\begin{itemize}
\item $\ctrv= (\ctr_1, \ldots, \ctr_k)$ is a vector of $k$ counter variables,
\item $\vars$ is a finite set of integer variables,
\item $\states$ is a finite set of states,
\item $\tr \subseteq \states\times \cc_k(X)\times \ic(X)\times\states\times\Zb^k$ is a transition relation,
%\item
\item $\init \in \states$ is the initial state.

\end{itemize}
% $\ctrv= (\ctr_1, \ldots, \ctr_k)$ is a vector of $k$ counter variables,
%$\vars$ is a finite set of integer variables,
%$\states$ is a finite set of states,
%$\tr \subseteq \states\times \cc_k(X)\times \ic(X)\times\states\times\Zb^k$ is a transition relation,
%%\item
%and $\init \in \states$ is the initial state.
\end{definition}

A transition $(q_1,\alpha,\beta,q_2,\vec{\kappa}) \in \delta$ moves the SCM from state $q_1$ to $q_2$ if the counters satisfy the constraint $\alpha$ and the  inspected array cell satisfies $\beta$; the counters are incremented by $\vec{\kappa}$, and the machine moves to the next cell. A machine is called \emph{deterministic} if $\tr$ is functional.
%, and \emph{one-way} if the tape head never moves back. 
A counter machine makes a \emph{reversal} if it makes an alternation between non-increasing and non-decreasing some counter.  A machine is
\emph{reversal-bounded} if there exists a constant $c\geq 0$ such
that on all accepting runs every counter makes at most $c$ reversal.

%%\subsection{Symbolic Counter Machines}
%We extend counter machines to symbolic counter machines (SCMs), which accept
%sequences of integers.  We denote the symbolic value of an array cell by a
%special integer variable $\invar$.  Given a set $X$ of integer
%variables (where $\invar\not\in X$), an \emph{atomic input constraint}
%is of the form $\invar \approx c$ or $\invar \approx x$, where
%$c\in \Nb$ and $x\in X$, and $\approx \; \in\{<,\le,>,\ge,=, \ne\}$.  An
%\emph{input constraint} is either a conjunction of $n\geq 1$ atomic input
%constraints, or the formula $\true$; we denote by
%$\ic(X)$ the set of all input constraints over variables in $X$.
%

\begin{definition}
\label{def:fscm}
  We define the translation of a functional constant $f$ of sort $\fsort{m}$, occurring in a formula $\phi$, as an SCM  $\scm(f) = (\ctrv, \vars, \states, \tr, \init)$. Let $G = \tuple{S,E, \gamma}$ be the edge-labeled graph for $f$ as defined in Section~\ref{afl-syntax}. Then 
$\ctrv = \{\idx,\cnt_1,\ldots,\cnt_{m-1}\}$, $\vars$ are fresh free variables for each integer term $T$ in $f$, $\states = S$, $\init = 0$, and for each edge $(s_1,s_2) \in E$, $\delta$ contains a transition from $s_1$ to $s_2$ labeled with a conjunction of all constraints labeling the edge. For each integer term $T$ in $f$ and the corresponding variable $x \in X$, we replace $T$ by $x$ in $f$, and add the assertion $(x=T)$ as a conjunction at the outermost level of $\phi$. Due to the constraint on $G$,  we have that $\scm(f)$ is reversal-bounded.
\end{definition}

%\begin{definition}
%\label{def:scm}
%  A $k$-\emph{symbolic counter machine} is a tuple
%  $\scm = (\ctrv, \vars, \states, \tr, \init)$, where
%\begin{itemize}
%\item $\ctrv= (\ctr_1, \ldots, \ctr_k)$ is a vector of $k$ counter variables,
%\item $\vars$ is a finite set of integer variables, such that $\invar\not\in X$
%\item $\states$ is a finite set of states,
%\item $\tr$ is a transition relation, which is a finite subset of 
%  \[\states\times \cc_k(X)\times \ic(X)\times\states\times\Zb^k\]
%\item $\init \in \states$ is the initial state.
%\end{itemize}
%\end{definition}

% \begin{definition}
% \label{def:scm}
% A symbolic counter machine (SCM)
% $\scm = (\ctrv, \states, \alphabet^s, \tr, \init)$ is a counter
% machine such that $\alphabet^s = \{\psi_1, \ldots, \psi_n\}$, where
% each $\psi_i$ is a finite subset of $\sc$.
% \end{definition}

Thus, we can translate a \Fold\ term into an SCM. A \emph{parallel composition} of SCMs captures the scenario when several \Fold{s} operate over the same array.

\begin{definition}
\label{def:pc}
The parallel composition (product) of two SCMs $\scm_1$ and $\scm_2$, where
$\scm_i = (\ctrv_i, \vars_i, \states_i, \tr_i, \init_i)$, is an
SCM $\scm = (\ctrv, \vars, \states, \tr, \init)$ such that:
\begin{itemize}
\item $\ctrv = \ctrv_1\ctrv_2$,
\item $\vars = \vars_1 \cup \vars_2$,
\item $\states = \states_1\times \states_2$,
\item for each pair of transitions
  $(q_i, \alpha_i, \beta_i, p_i, \vec{w}_i)\in \tr_i$, where $i=1..2$, there is the transition $\big((q_1, q_2),\alpha_1\land \alpha_2, \beta_1 \land \beta_2, (p_1, p_2),\vec{w}_1\vec{w}_2 \big)\in \tr$, which are the only transitions in $\tr$,
\item $\init = (\init_1, \init_2)$.
\end{itemize}

%
%$\ctrv = \ctrv_1\ctrv_2$,
%$\vars = \vars_1 \cup \vars_2$,
%$\states = \states_1\times \states_2$,
%$\init = (\init_1, \init_2)$,
%and 
%for each pair of transitions
%  $(q_i, \alpha_i, \beta_i, p_i, \vec{w}_i)\in \tr_i$, where $i=1..2$, there is the transition $\big((q_1, q_2),\alpha_1\land \alpha_2, \beta_1 \land \beta_2, (p_1, p_2),\vec{w}_1\vec{w}_2 \big)\in \tr$, there are no other transitions in $\tr$.

%$\tr$ is defined as follows: if
%  $(q_i, \phi_i, \psi_i, q_i, \vec{w}_i)\in \tr_i$ ($i=1..2$), then
%  \[((q_1, q_2),\phi_1\land \phi_2, \psi_1 \land \psi_2, (p_1, p_2),\vec{w}_1\vec{w}_2)\in \tr.\]

\end{definition}

%\begin{definition}
%\label{def:pc}
%The parallel composition of two SCMs $\scm_1$ and $\scm_2$, where
%$\scm_i = (\ctrv_i, \vars_i, \states_i, \tr_i, \init_i)$ is an
%SCM $\scm = (\ctrv, \vars, \states, \tr, \init)$, such that:
%\begin{itemize}
%\item $\ctrv = \ctrv_1\ctrv_2$,
%\item $\vars = \vars_1 \cup \vars_2$,
%\item $\states = \states_1\times \states_2$,
%\item $\tr$ is defined as follows: if
%  $(q_i, \phi_i, \psi_i, q_i, \vec{w}_i)\in \tr_i$ ($i=1..2$), then
%  \[((q_1, q_2),\phi_1\land \phi_2, \psi_1 \land \psi_2, (p_1, p_2),\vec{w}_1\vec{w}_2)\in \tr.\]
%\item $\init = (\init_1, \init_2)$.
%\end{itemize}  
%\end{definition}

One of the fundamental questions that can be asked about a logic concerns the size
of its models. The following lemma shows that models of bounded size are enough to check satisfiability of an AFL formula.

\begin{lemma}[Small model property]
\label{th:small}
There exists a constant $c\in \Nb$, such that an $AFL$ formula $\phi$ is satisfiable iff there exists a model $\model$ such that 
%\begin{itemize}
%\item 
a) for each integer variable $x$ in $\phi$, $\model$ maps $x$ to an integer $\leq 2^{|\phi|^c}$,
%\item
and b) for each array variable in $\phi$, $\model$ maps the variable to a
  sequence of $\leq 2^{|\phi|^c}$ integers, where each integer is $\leq 2^{|\phi|^c}$.
%\end{itemize}
\end{lemma}
\begin{proof}[sketch; see the appendix for the full proof]
  One direction of the proof is trivial.  For the other direction,
  assume that $\phi$ has a model $\model$.  We construct a formula
  $\psi$ that is a conjunction of all atomic formulas of $\phi$: in positive polarity
  if $\model$ satisfies the atomic formula, and in negative
  otherwise.  Let $s=|\psi|$, and note that $s\leq 3|\phi|$.   We observe that 1) $\model$ is a model of $\psi$, 2) every
  model of $\psi$ is a model of $\phi$. In the remaining part of the
  proof we show that $\psi$ has a small model, and as a consequence so
  does $\phi$.

  Let $a$ be some array in $\psi$.  We translate each \Fold\ term over
  $a$ to an SCM $\scm_j$ as in Definition~\ref{def:fscm}; let SCM
  $\scm$ be the product of all $\scm_j$.  We extend the technique of
  \cite{Ibarra81} to show that there exist a sufficiently short run of
  $\scm$. Under the interpretation $\model$, all variables in counter
  constraints become constants.  Let $\vec{c}=(c_1, \ldots c_n)$ be a
  non-decreasing vector of constants that appear in the counter
  constraints of $\scm$ after fixing $\model$. Vector $\vec{c}$ gives
  rise to the set of regions
  \[\mathcal{R}= \{[0,c_1],[c_1,c_1],[c_1+1,c_2-1],[c_2,c_2],\ldots,[c_l,\infty]\}.\]
  The size of $\mathcal{R}$ is at most $2\dim(\vec{c})+1 \leq 3s$.  A
  mode of $\scm$ is a tuple in $\mathcal{R}^k$ that describes the
  region of each counter.  Let us observe that each counter can
  traverse at most $|\mathcal{R}|$ modes before it makes an
  additional reversal. Thus, $\scm$ in any run can traverse at most
  $max=r\cdot k\cdot|\mathcal{R}|\leq \mathcal{O}(s^3)$ different
  modes.
  
  We take some accepting run $Tr$ of $\scm$ that traverses at most
  $max$ modes, and partition sequences of transition in $Tr$ into
  equivalence classes.  We create an integer linear program LP
  that encodes an accepting run of $\scm$ that traverses at most $max$
  modes, as well as all non-fold constraints of $\psi$.  The variables
  of LP correspond to 1) the integer variables of $\psi$, 2) the
  counter values of $\scm$, 3) the number of times sequences from each equivalence class are taken, and 4) the solutions to each input constraint of
  $\scm$.

  We show that LP has a  solution $p$, where each variable is
  at most $\leq 2^{|\phi|^c}$, for a fixed $c$.
  We use $p$ to construct  a small model for  $\psi$.  
  From $p$ we immediately get interpretation for integer variables of
  $\psi$.  Solution $p$ implies that there is an accepting run of
  $\scm$ of length at most $\leq 2^{|\phi|^c}$, which also gives a
  bound on the length of the input array. Finally, for every array
  cell we may use a solution to the specific input constraint.  \hfill $\qed$
\end{proof}

As a consequence of Lemma~\ref{th:small} we obtain a result on the
complexity of AFL satisfiability checking.

\begin{theorem}
\label{th:comp}
The satisfiability problem of AFL is $\PSPACE$-complete.
\end{theorem}

\begin{proof}
  \emph{Membership.} By Lemma~\ref{th:small}, if an AFL formula
  $\phi$ is satisfiable, then it has a model where integer variables
  have value $\leq 2^{|\phi|^c}$, and arrays have  length
  $\leq 2^{|\phi|^c}$, and where each array cell stores a number
  $\leq 2^{|\phi|^c}$.  A non-deterministic Turing machine can use a
  polynomial number of bits to: 1) guess the value of integer
  variables and store them using $|\phi|^c$ bits, 2) guess one-by-one
  the value of at most $2^{|\phi|^c}$ array cells, and simulate the
  \Fold{s}. The Turing machine needs $|\phi|^c$ bits for counting the
  number of simulated cells.  The maximum constant used in a counter
  increment can be at most $2^{|\phi|}$.  Then, the maximal value a \Fold\
  counter can store after traversing the array is at most
  $2^{|\phi|^{c+1}}$, therefore polynomial space is also
  sufficient to simulate \Fold\ counters.

  \emph{Hardness.} We reduce from the emptiness problem for
  intersection of deterministic finite automata, which
  is $\PSPACE$-complete \cite{Kozen77}. We are given a sequence $A_1,\ldots, A_n$ of
  deterministic finite automata, where each automaton $A_i$ accepts the
  language $\mathcal{L}(A_i)$.  The problem is to decide whether
  $\bigcap_{i=1}^{n} \mathcal{L}(A_i)\neq\emptyset$.  We simulate
  automata $A_i$ with a \Fold\ expression $\Fold_{a}^i$ over a single
  counter, where input constraints correspond to the alphabet symbols
  of the automata.  The expression $\Fold_{a}^i$ returns an even
  number on array $a$ if and only if the interpretation of $a$
  represents a word in $\mathcal{L}(A_i)$.  To check emptiness of
  the automata intersection, it is enough to check whether there exists an
  array such that all folds $\Fold_{a}^1,\ldots,\Fold_{a}^n$ return an
  even number.  The reduction can be done in polynomial time.  $\qed$
\end{proof}

\subsection{Undecidable Extensions}
\label{sec:undecidability}

%In this section w
We show that two natural extensions to our logic lead to undecidability.

\begin{theorem}
\label{hilbert}
Array Folds Logic with an $\exists^*\forall^*$ quantifier prefix is undecidable.
\end{theorem}
\begin{proof}
We prove by a reduction from Hilbert's Tenth Problem \cite{Hilbert}; since addition is already in the logic, we only show how to encode multiplication. The following $\exists^*\forall^*$ AFL formula has a model iff array $a$ is a repetition of $z$ segments, and each segment is of length $y$ and has the shape $00...01$; thus, it asserts that $x = y\cdot z$:
%\vspace{-0.5mm}
\begin{gather*}
 |a|=x \fand \FoldPI{a}{\vect{0;0}}{\vect{\elm=0 \Then \Skip ;\elm=1 \Then \cnt_1\pp}} = \Vect{|a|;z} \fand\\
 \forall j\,.\, 0\le j <|a| \implies \FoldPI{a}{\vect{j;0}}{\vect{\idx\,\le\,j+y \fand \elm=0 \;\Rightarrow\; \Skip;\idx\,\le\,j+y \fand \elm=1 \;\Rightarrow\; \cnt_1\pp}} = \Vect{*;1}\qquad \qed
\end{gather*}

In \cite{Buchi88}, the following is proved about the theory of  concatenation:

\begin{theorem}[\cite{Buchi88}, Corollary 4; see also \cite{Furia10}, Proposition 1] \\
Solvability of equations in the theory $\tuple{\{1,2\}^*,e,\circ,\var{Lg}_1,\var{Lg}_2}$, where $\var{Lg}_p(x) \equiv \{y \in p^* \;|\; y \text{ has the same number of }p\text{'s as } x\}$,  is undecidable.
\end{theorem}

\begin{corollary}
Extension of AFL with concatenation operator $\circ$ is undecidable.
\end{corollary}
\begin{proof}
For an array $x$, we can define another array $\var{Lg}_1(x)$ in AFL as follows:
%\vspace{-0.5mm}
\[
\Vect{|x|;|\var{Lg}_1(x)|} = \FoldPI{x}{\vect{0;0}}{\vect{\elm = 1 \Then \cnt_1\pp;\elm \ne 1 \Then \Skip}} \fand \Vect{|\var{Lg}_1(x)|} = \FoldPI{\var{Lg}_1(x)}{\vect{0}}{\vect{\elm = 1 \Then  \Skip}} 
\qquad \qed
\]
\end{proof}

%\begin{gather*}
% a[0]=1 \;\land \; |a|=x\\
% \forall i\,.\, i<|a|-y \implies \Fold\!\bm{\vect{j=0;c=0}}\!\bm{\vect{j++;a=1 \;\Rightarrow\; c++}} = \bm{\vect{y;1}} \;\land \\
% \Fold\!\bm{\vect{c=0}}\!\bm{\vect{a=1 \;\Rightarrow\; c++}} = \bm{\vect{z}}
%\end{gather*}
\end{proof}

%%% Local Variables:
%%% mode: latex
%%% TeX-master: "main"
%%% End:

\section{Decision Procedure}
\label{sec:procedure}
In Section~\ref{sec:complex} we described how a non-deterministic Turing
machine can decide AFL satisfiability in $\PSPACE$. Now we present
a deterministic procedure that translates AFL formulas to equisatisfiable 
quantifier-free Presburger formulas. As a
consequence of the procedure, we show that under certain restrictions
satisfiability of AFL is $\NP$-complete.

\textbf{Deterministic procedure} We are given an AFL formula $\phi$
such that there are at most $m$ \Fold{s} over each array; clearly
$m$ can be at most $|\phi|$. We translate $\phi$ to the quantifier-free Presburger formula
$\psi = \psi_{n} \land \psi_{e} \land \psi_{l}$. 
For the procedure we assume that there exists a fixed order $x_1\leq \cdots \leq x_n$ on variables that appear in the counter constraints.

\emph{Formula $\psi_n$.} The formula $\psi_n$ is the part of $\phi$ that does not contain \Fold{s}. 

\emph{Formula $\psi_e$.} For an array $a_j$ in $\phi$, let
$F_j = \{\Fold^1_{a},\ldots,\Fold^m_{a}\}$ be the set of \Fold{s} in $\phi$
over $a_i$. We translate each $\Fold^i_{a}\in F_j$ to a symbolic
counter machine $\scm^i_j$.  Each $\scm^i_j$ has at most $|\phi|$
transitions, and the sum of the counters and the number of reversals among all $\scm^i_j$
is at most $|\phi|$. Next, we construct the symbolic counter machine
$\scm_j$ as the product of all machines $\scm^i_j$.  The machine
$\scm_j$ has at most $k=|\phi|$ counters, $t=|\phi|^m$ transitions and
makes at most $r=|\phi|$ reversals.

We translate the reachability problem of $\scm_j$ to the quantifier-free
Presburger formula $\psi_e^j$ by applying an extension of the method
described in \cite{HagueL11}.  In formula $\psi_e^j$, two
configurations of $\scm_j$ are described symbolically: initial
$\conf$, and final $\conf'$.  The formula $\psi_e^j$ is satisfiable iff there is an array $a_j$ such that $\scm_j$ reaches $\conf'$ from $\conf$ on reading $a_j$. The formula $\psi_e$ is the
conjunction $\psi_e^j$ for all arrays $a_j$.

The formula $\psi_e^j$ consists of two parts
$\psi_e^j = \psi_p^j \land \psi_c^j$.  For simplicity we assume that
the counter constrains of $\scm_j$ are defined only over variables
$\{x_1,\cdots,x_n\}$.  By assumption, there is a fixed order
$x_1\leq \ldots \leq x_n$, which gives rise a to the set of
$\leq 2|\phi|+1$ regions
$\mathcal{R}=
\{[0,x_1],[x_1,x_1],[x_1+1,x_2-1],\cdots,[c_l,\infty]\}.$ As an optimization, we construct regions separately
for each counter, which allows us to obtain a tighter bound on the
number of regions that need to be encoded.

Each counter may traverse at most $|\mathcal{R}|$ regions before it
makes a reversal, so an accepting computation of $\scm_j$ traverses at
most $max = r\cdot k\cdot|\mathcal{R}| = \mathcal{O}(|\phi|^3)$ modes.
We construct an NFA $\mathcal{A}_j$ by making $max$ copies of the
control-flow structure of $\scm_j$.  Every run of $\mathcal{A}_j$
gives a correct sequence of states in $\scm_j$, but may violate
counter constraints.  By using the procedure of
\cite{seidl2004counting} we can encode the Parikh image of
$\mathcal{A}_j$ as the formula $\psi_j^p$ that is polynomial in the
size of $\mathcal{A}$.  Similar to \cite{HagueL11}, the formula
$\psi^j_c$ puts additional constraints on the Parikh image to ensure
that by executing the transitions of $\mathcal{A}_j$ we obtain counter
values that satisfy the counter constraints of $\scm_j$.

The size $\psi_e^j$ is of the order
$\mathcal{O}(|\phi|^3 t)=\mathcal{O}(|\phi|^{m+3})$.  The formula
$\psi_e$ is the conjunction of formulas $\psi_e^j$ for each array
$a_j$.  There can be at most $|\phi|$ arrays, so the size of $\psi_e$
is $\mathcal{O}(|\phi|^{m+4})$.

\emph{Formula $\psi_l$.} Finally, formula $\psi_l$ links the initial
and final configurations in $\psi_e$ to the variables in $\psi_p$.

\emph{Formula size.} The size of the formula $\psi$ is
$\mathcal{O}(|\phi|^{m+4})$.  By keeping $m$ constant, the encoding
size is polynomial in the size of the AFL formula $\phi$.

\textbf{Restricted fragment of AFL} We write $m$-AFL for formulas that
have at most $m$ \Fold\ expressions per array. As a consequence of the
deterministic decision procedure, restriction on $m$ reduces the
complexity of deciding satisfiability.

\begin{lemma}
\label{th:rescomp}
The $m$-AFL satisfiability problem, for a fixed $m$, is $\NP$-complete.
%For a fixed constant $m$, the satisfiability problem of an $m$-AFL formula is $\NP$-complete.
\end{lemma}
\begin{proof}
  Membership follows from the decision procedure above.
  For hardness observe that any quantifier-free Presburger formula is an $0$-AFL formula.
$\qed$
\end{proof}

\textbf{Model generation} Given a Presburger encoding $\psi$ of an AFL
formula $\phi$, we may use the solution to $\psi$ to generate a model
of $\phi$. The solution to $\psi$ immediately gives us interpretation
for the integer variables in $\phi$. To obtain an interpretation for the array
variables in $\phi$, we observe that folds are implicitly encoded in
$\psi$ as counter machines, and that the solution to $\psi$ describes
the Parikh vector for each machine. We use the method of
\cite{seidl2004counting} to get a concrete sequence of transitions in
each counter machine that produces the specific Parikh vector. We construct a multigraph by
repeating each transition in $\mathcal{A}_j$ according to its Parikh image, and
then find an Eulerian path in the multigraph.  From the sequence of
transitions in counter machines, and the interpretation of input
constraints in $\psi$ we obtain an interpretation for the arrays in $\phi$.

%%% Local Variables:
%%% mode: latex
%%% TeX-master: "main"
%%% End:

%\section{Applications in Testing and Verification}
%\label{sec:applications}

%\section{Extensions}
%\label{sec:extensions}

\section{Experiments}
\label{sec:experiments}

We implemented the decision procedure described in
Section~\ref{sec:procedure} in a prototype tool \textsc{AFolder}; the
tool is available at \cite{tool}. The tool is written in \texttt{C++}
and uses \textsc{Z3}~\cite{z3} as the solver for Presburger
formulas. We evaluated our decision procedure on a number of testing
and verification tasks described below.

% \subsection{Optimizations}
% \label{sec:app_opt}

% We evaluated our
% decision procedure on the Markdown program from
% Section~\ref{sec:example}, a number of  verification task from
% \textsc{SvComp}~\cite{svcomp15}, a piece of code from the
% \textsc{Linux} kernel \cite{perf-numa}, and the histogram example from 

%\subsection{Results}
The experimental results are shown in Table \ref{tab:results}; all
experiments were performed on a Ubuntu-14.04 64-bit machine running on
an Intel Core i5-2540M CPU of 2.60GHz.  For every example we report
the size $|\phi|$ of the AFL formula measured as ``the number of
logical operators'' + ``the number of branches in folds.''  The table
also shows the number of \Fold\ expressions in a formula, and the
maximum number of folds per array (MFPA). Next, we report the time for
translating the problem to a Presburger formula, the time for solving the
formula, and whether the formula is satisfiable. If this is the case, we report the length of a satisfying array generated
by our tool; in case of several arrays, we show the longest.

\begin{table}[tp]
\caption{Experimental results for \textsc{AFolder}.}
% MFPA is the maximum number of folds per array. }
\label{tab:results}
\centering
\resizebox{\columnwidth}{!}{
\begin{tabular}{l|cccccccc} 
\hline
\multicolumn{1}{c|}{Example} 		& \quad\;$|\phi|$\;\quad &  \quad folds \quad& \quad MFPA\quad &\quad transl.\ time\quad&\quad solving time\quad & \quad result \quad& array length\\
\hline
Markdown(1)		&  62   & 6 	& 3 & $<1s$	& $<1s$ & sat & $8$ \\
Markdown(2)		&  69   & 7 	& 4 & $1s$	& $<1s$ & sat & $14$ \\
Markdown(3)		&  76   & 8 	& 5 & $1.3s$	& $79$s & sat & $17$ \\
\hline
perf\_bench\_numa(10)	& 93	& 10	& 1 & $<1$s	& $<1$s	 & sat & $100$\\
perf\_bench\_numa(20)	& 183	& 20	& 1 &$<1$s	&$<1$s	 & sat & $100$\\
perf\_bench\_numa(40)	& 363	& 40	& 1 &$<1$s	&$<1$s	 & sat & $100$\\
\hline
standard\_minInArray 	& 10	&  3	& 3 & $<1$s	&$<1$s	& unsat & - \\
linear\_sea.ch\_true 	& 13	&  3	& 3 & $<1$s	& $<1$s	& unsat & - \\
array\_call3		& 11	&  2	& 3 & $<1$s	&$<1$s	& unsat & - \\
standard\_sentinel	& 14	&  3	& 3 & $<1$s	&$<1$s	& unsat & - \\
standard\_find		& 11	&  3	& 3 & $<1$s	&$<1$s	& unsat & - \\
  standard\_vararg	& 11	&  3	& 3 & $<1$s	& $<1$s	& unsat & - \\
\hline
histogram(8)		& 58	& 8  	& 8  & $<1$s 	&$1.3$s & sat	& 9 	\\
histogram(9)		& 65	& 9  	& 9  & $<1$s 	&$6.9$s & sat	& 10 	\\
histogram(10)		& 72	& 10 	& 10 & $2$s	&$55$s  & sat	& 11 	\\
histogram(11)		& 79	& 11 	& 11 & $8$s	& $368$s & sat	& 12 	\\
histogram\_unsat(11)	& 80	& 11 	& 11 & $9$s	&$19$s & unsat	& -	\\
\hline
\end{tabular}}
%\vspace*{-2mm}
\end{table}

\textbf{Markdown} This program is described in
Section~\ref{sec:example}. The experiments are parametrized by the required number $n$ of columns in the input.

\textbf{perf\_bench\_numa} This example is part of a benchmark program
for non-uniform memory access (NUMA) \cite{perf-numa}.  The program
maintains a list of threads, and for each thread a separate array of size
$100$ that describes processors assigned to the thread. The data is
processed in a nested loop: the outer loop iterates over threads, and
the inner loop counts the number of assigned processors. 
The outer
loop also maintains the minimum, and maximum number of processors
assigned to any thread.  
We model a testing scenario like in
 Section~\ref{sec:example}, where a symbolic execution tool unrolls the
 outer loop $n$ times, and the inner loop is summarized by a fold
 expression. 
The testing goal is to provide a valid processor mapping such that
each thread is assigned to exactly one processor.
 In Table \ref{tab:results} we show results for this benchmark parametrized by
the number $n$ of threads.  The example scales well, since there a
 single fold per each processor array (see Lemma~\ref{th:rescomp}).

\textbf{\textsc{SV-COMP}} Examples ``standard\_minInArray'' to
``standard\_vararg'' are taken from the \textsc{SV-COMP} benchmarks
suite~\cite{svcomp15}.  They model simple verification problems for
loops, such as finding the position of an element in array,
finding the minimum, or counting the number of positive elements. We
model these programs as formulas that are unsatisfiable if the
program is safe.  Although the programs are simple, most verification
tools competing in \textsc{SV-COMP} fail to prove their safety.

\textbf{histogram} We performed experiments on the histogram example 
in Section~\ref{sec:express}, parametrized by the number of range
values. We observe that solving time grows rapidly with the number of
folds.  Example ``histogram\_unsat'' is an unsatisfiable variation that
requires two different counts in the same range.

%%% Local Variables:
%%% mode: latex
%%% TeX-master: "main"
%%% End:

\section{Conclusion and Future Work}
\label{sec:conclusions}
We presented Array Folds Logic (AFL), which extends the quantifier-free theory of arrays with \emph{folding}, a well-known concept from functional languages. The extension allows us to express counting properties, occurring frequently in real-life programs. Additionally, AFL is able to concisely summarize loops with  internal branching and counting over arrays. We have analyzed the complexity of satisfiability checking for AFL formulas, and presented an efficient decision procedure via an encoding to the quantifier-free Presburger arithmetic. Finally, we have implemented a tool called \textsc{AFolder}, which efficiently discharges AFL proof obligations, and demonstrated its practical applicability on numerous examples.

For the future work, we plan to investigate
% possible extensions to the expressive power of the logic (such as the \emph{accumulation} properties), as well as on 
 possible combinations with other decidable fragments of the theory of arrays (to allow some restricted form of quantifier alternation). We also plan to automate the generation of proof obligations and the summarization of loops, and want to improve the efficiency of our decision procedure by implementing suitable optimizations and heuristics.

%In this paper we presented Array Folds Logic that extends the
%quantifier-free theory of arrays by the fold expression. Folding, a
%concept well-known from functional programming, allows our logic to
%express counting properties that cannot be handled by other array
%logics.  Our logic is especially useful for dealing with loops that
%traverse arrays and perform counting --- a task task is notoriously
%difficult for verification and testing tool.
%
%For future work, we plan to investigate automatic translation of program loops into fold expressions.

%%% Local Variables:
%%% mode: latex
%%% TeX-master: "main"
%%% End:

%\subsubsection*{Acknowledgments}

%\newpage
\bibliographystyle{abbrv}
\bibliography{ref}

\clearpage
\appendix

\section{Proof of Lemma~\ref{th:small}}
\label{sec:app_proofs}

\label{ref:proofsmall}
\begin{proof}
  One direction of the proof is trivial.

  For the other direction, assume that $\phi$ has a model $\model$.
  Let $X$ be the set of variables in $\phi$.
  W.l.o.g. we assume that all folds are of the form
  \[\Vect{out_1;\cdots;out_n}~=~\Fold_{a}\: \Vect{in_1;\cdots;in_n}\: F^m\]
  where $out_1,\ldots,out_n,in_1,\ldots,in_n$ are integer variables.

  From the model $\model$ of $\phi$ we build a conjunction $\psi$ of
  literals in the following way: for every atomic formula $\gamma$
  of the form $T=T, V^m=V^m$, or $T\leq T$ in $\phi$ we add a conjunct
  $\gamma$ to $\psi$ if $\model$ satisfies $\gamma$, and we add a
  conjunct $\neg \gamma$ otherwise.  Observe that $\model$ is a
  model of $\psi$ and every model of $\psi$ is a model of $\phi$. In
  the remaining part of the proof we show that $\psi$ has a small
  model.

  Let $s=|\psi|$, and note that $s\leq 3|\phi|$.  Moreover, we
  write $\psi = \psi_{f} \land \psi_{nf}$, where $\psi_{f}$ contains
  only literals with folds, and $\psi_{nf}$ contains only literals
  without folds.

  Let us assume that folds are over the same array $a$; we will later
  deal with this restriction.  Let
  $F = \{\Fold^1_{a},\ldots,\Fold^n_{a}\}$ be a set of folds in
  $\psi$ over the array $a$.  We translate each fold $\Fold^i_{a}\in F$ to a
  symbolic counter machine $\scm_i$.  Each $\scm_i$ has at most $s$
  transitions, and the sum of counters and reversal among all $\scm_i$
  is at most $s$.  Next, we create the product
  $\scm=(\ctrv, \states, \alphabet, \tr, \init)$ of all the folds in
  $F$.  The product counter machine $\scm$ has at most $k=s$ counters,
  $s^s$ states and transitions, and makes at most $r=s$ reversals.

  % There is mapping from the inital values of counters and the
  % variables in $Out$, and similarly, from the output value of
  % counters to the variables in $In$.  The
  % Model $\model$ assigns vectors $\vec{w}_i, \vec{w}_o\in \Nb^k$ of
  % initial values, and output values to the counters of $\scm$.

  In the following part of the proof, we extend the technique of
  \cite{Ibarra81} to show that there exist a sufficiently short run of
  $\scm$. Under the interpretation $\model$, all variables in counter
  constraints become constants.  Let $\vec{c}=(c_1, \ldots c_n)$ be a
  non-decreasing vector of constants that appear in the counter
  constraints of $\scm$ after fixing $\model$. Vector $\vec{c}$ gives
  rise to a set of regions
  \[\mathcal{R}= \{[0,c_1],[c_1,c_1],[c_1+1,c_2-1],[c_2,c_2],\ldots,[c_l,\infty]\}.\]
  The size of $\mathcal{R}$ is at most $2\dim(\vec{c})+1 \leq 3s$.  A
  mode of $\scm$ is a tuple in $\mathcal{R}^k$ that describes the
  region of each counter.  Let us observe that each counter can
  traverse at most $|\mathcal{R}|$ modes before it makes an
  additional reversal. Thus, $\scm$ in any run can traverse at most
  $max = r\cdot k\cdot|\mathcal{R}|\leq \mathcal{O}(s^3)$ different
  modes.

  Let $\vec{w}^{in}, \vec{w}^{out}\in \Nb^k$ be the vectors of the initial and
  final values of counters of $\scm$; these vectors are given by
  interpretation $\model$ of the initial and output variables of folds.
  We know that $\scm$ on input $\model(a)$ has a run
  $Tr=\delta_1,\ldots, \delta_n$, where $\delta_i\in \tr$, from an initial configuration
  $\conf_0=(\init, \vec{w}_{in})$ to a final configuration
  $\conf_n=(\cdot, \vec{w}_{out})$.

  We partition $Tr$ into sub-sequences $Tr_1, \ldots, Tr_{max}$, such
  that all transitions in $Tr_i$ are fired from a configuration in mode
  $i$, and for $i<max$ the last transition in $Tr_i$ leads to a
  configuration in mode $i+1$.  Let us look into some sub-run
  $Tr_i=\delta_l\ldots\delta_{l'}$. For each transition $\delta\in \tr$ we
  mark one occurrence of $\delta$ in $Tr_i$, provided that such
  transition occurs. In this way, we mark at most $|\tr|\leq s^{s}$
  transitions in $Tr_i$.

  Next, we identify sub-sequences $\rho=\delta_m,\ldots,\delta_n$ of $Tr_i$,
  such that 1) $m<n$, 2) $\delta_m=\delta_n$, 3) all transitions in
  $\rho$ are unmarked, and 4) all transitions
  $\delta_{m+1},\ldots,\delta_{n-1}$ are distinct. Observe that by
  deleting transitions $\delta_2,\ldots,\delta_l$ from $Tr_i$ we obtain
  a valid sequence of transition that ends in the same state, but may
  lead to different counter values. Let $\bar{Tr}_i$ be the sequence of
  transitions that results by repeatedly deleting all such sequences
  from $Tr_i$, and let $S_i$ be the multi-set of sequences deleted in
  such way. Since marked configurations remain in $\bar{Tr}_i$ (and
  there are at most $|\tr|$ of them), and there can be at most $|\tr|$
  remaining transition between any two marked configuration in
  $\bar{Tr}_i$, therefore $\bar{Tr}_i$ has at most
  $|\tr|(1+|\tr|)=\mathcal{O}(s^{3s})$ transitions.

  For each $S_i$ we define an equivalence relation $S^=_i$ on the
  deleted transitions as follows: two deleted sequences are equivalent
  if they 1) have the same starting transitions, 2) have the same end
  transitions, and 3) add/subtract the same value for each
  counter. Note that in $Tr_i$ a deleted sub-sequence can be
  substituted by an equivalent one, without changing the final
  configuration, i.e. the state and the value of counters remain the
  same.  Let $u$ be the maximum constant which is added/subtracted in
  a counter update; $u$ is given in binary, so $u \leq 2^{s}$.  There
  may be at most $|\tr|$ transitions in a deleted sequence, so the
  total net effect on a single counter may be at most
  $u|\tr|=s^s2^s$. The number of equivalence classes in $S^=_i$ is
  $|S^+_i| \leq u|\tr|^2 \leq s^{3s}2^s\leq \mathcal{O}(s^{4s})$.

  We construct an integer linear program $LP$, which expresses solutions to
  $\psi$ that lead to some run $Tr^*$ of $\scm$ such that 1)
  $Tr^*$ goes through the same modes as $Tr$, 2) $Tr^*$ has the same
  net effect on each counter, 3) $Tr^*$ is possibly shorter than $Tr$.
  The linear program has five parts
  $LP = LP_1 \land \ldots \land LP_5$.

  $\mathbf{LP_1}$ In $LP_1$ we specify constraints for the counter values.
  For every counter $1\leq j\leq k$ and mode $1\leq i\leq max$ we
  create variables $w_{i,j}$ and $z_{i,j}$ which describe the value of
  counter $j$ at the start and the end of mode $i$.  Thus, $z_{max,j}$
  is the value of counter $j$ after executing $Tr^*$, and $w_{1,j}$ is
  the initial value of the computation.
  
  For $1\leq j\leq k,1\leq i\leq max$, part $LP_1$ contains the following equations
  \[z_{i,j} = w_{i,j} + \bar{b}_{i,j} +
  \sum_{m=1}^{|S^=_i|}b_{i,j,m}~y_{i,m},\]
  where $\bar{b}_{i,j}\in \Nb$ is the net effect of $\bar{Tr}_i$ on
  counter $j$, $b_{i,j,m}\in \Nb$ is the net effect of the sequence in
  equivalence class $m$ on counter $j$, and variable $y_{i,m}$ denotes
  the number of times that the sequence from class $m$ occurs in mode
  $i$.

  Additionally, for $1\leq i< max-1$, we add to $LP_1$ the
  constraint:
  \[ w_{i+1,j} = z_{i,j} + b^+_{i,j}, \] % check this
  where $b^+_{i,j}$ is the effect of the last transition in $\bar{Tr}_i$
  on counter $j$.
  We denote by $Z,W,Y$ the sets of variables
  $z,w,y$, respectively.

  $\mathbf{LP_2}$ In $LP_2$ we specify the constrain for
  $\psi_{nf}$, which is the part of $\psi$ that does not contain
  folds.We create a copy $X'$ of all the variables in $X$. The linear program
  contains a formula
  \[ LP_2 = \psi_{nf}[X/X']. \]
  
  $\mathbf{LP_3}$ We link the start and end value of counters with
  variables in $X'$.  For a counter $1\leq j \leq k$ we add the
  following constraints to $LP_3$
  \[ z_{max,j} = out'_k \land w_{i,j} = in'_k, \]
  where $out_k, in_k$ are the variables that specify in $\psi$ the initial
  and output values of counter $j$ of $\scm$.

  $\mathbf{LP_4}$ W.l.o.g.\ we assume that in the folds of $\psi$
  counters are only compared to variables. As a result of our
  translation, the constants in counter constraints of $\scm$ arise
  from interpretation $\model$ of some variables in $X$.  Thus, we add
  the constraint that counter values stay within their modes. Suppose
  in mode $i$, counter $j$ is in the region specified by the interval
  $[x_k, x_l]$, where $x_k, x_l\in X$.  We add the following
  constraints to $LP_4$:
  \[ x'_l \leq w_{i,j} \leq x'_k \quad\land\quad x'_l \leq z_{i,j} \leq x'_k.\]
  
  $\mathbf{LP_5}$ Finally, we need to ensure that all input constraints
  executed in $Tr$ are satisfiable. Thus, let $\mathcal{I}$ be the set
  of input constraints executed in $Tr$, where $|\mathcal{I}|\leq
  s$.
  For every input constraint $\psi_i\in \mathcal{I}$ we create an new
  variable $v_i$ and add to $LP_5$ a linear constraint over $v^i$ that
  corresponds to $\psi_i$.  We denote by $V$ the set of $v_i$
  variables.

  The linear program $LP$ can be expressed in the form
  $\vec{A}\vec{x} \leq \vec{b}$, where $\vec{x}$ is a vector of
  variables in $X'\cup Z\cup W\cup Y\cup V$, and and
  $\vec{A}, \vec{b}$ are a matrix, and vector of integers. 
  To determine the dimension of $\vec{x}$, observe that the cardinalities
  of $X', Z, W, V$ are polynomial in $s$, while
  \[|Y|\leq \sum_{i=1}^{max}|S^=_i| \leq \mathcal{O}(s^{k_0s}),\]
  where $k_0$ is a fixed constant.  Thus,   $\dim(\vec{x})\leq \mathcal{O}(s^{k_0s})$. Each of the constraints
  in $LP_1,\ldots, LP_2$ adds a number of constrains polynomial in
  $s$, so $\dim(\vec{b})\leq \mathcal{O}(s^{k_1})$, where $k_1$ is a
  fixed constant.  Observe that all entries in $\vec{A}$ and $\vec{b}$
  are at most of the size $2^s$ -- the largest number we can encode in
  $\psi$.  As a result the size of $[\vec{A}\vec{b}]$ is at most
  $\mathcal{O}(s^{sk_2})$, where $k_2$ is a fixed constant.

  We know that there exists a solution to $\vec{x}$ that satisfies the
  linear program: this solution can be obtained be assigning
  appropriate variables to the values of $Tr$ and $\model$.  By
  \cite{papa81} we know that there also exists a solution $p$ to
  $\vec{x}$ to assigns to every variable an element with absolute
  value is at most $2^{s^{k_3}}$, for some fixed constant $k_3$ (to
  satisfy the requirements of \cite{papa81}, we transforms $LP$
  program to the standard form
  $\vec{A'}\vec{x'} = \vec{b'}, \vec{x}\geq 0$ -- this can be done by
  doubling the number of constraints and variables.  The solutions to
  $\vec{x'}$ can be mapped to the solutions in $\vec{x}$).

  From the solution $p$ we can construct a model $\model'$ of
  $\psi$ in the two steps.  First, for integer variables in $X$,
  assign the values of $X'$ under $p$. Clearly, each variable gets a value
  with a representation that is $\mathcal{O}(s^{k_3})$ bits.

  Second, to get an assignment $\model'$ to array $a$ we create a
  computation $Tr^*$ that has the same effect as $Tr$.  To obtain
  $Tr^*_{i}$, in every reduced sub-sequence $\bar{Tr}_i$ place
  $p(y_{i,m})$ copies of a sequence in the equivalence class $m$ of
  $S^=_i$ after the ``marked transition rule identical to the last
  transition rule in the sequence'' (see \cite{Ibarra81}).  Then
  $Tr^* = Tr^*_{1},\ldots,Tr^*_{max}$, and the length of $Tr^*$ is in
  $\mathcal{O}(s^{sk_4})$, where $k_4$ is a fixed constant.  The
  assignment to array $a$ is a sequence of integers which satisfy the
  input constraints that run $Tr^*$ reads. For each input constraint
  $\psi_i$ read by $Tr^*$, we may use a solution $p$ to $v_i$ as a
  concrete value. The assignment to array $a$ is of length
  $\mathcal{O}(s^{sk_4})$, and each element can be represented in
  at most $\mathcal{O}(s^{k_3})$ bits.  To complete the proof, let us
  observe that there can be at most $|\psi|$ arrays in the formula
  $\psi$, so having multiple arrays does not change the order of
  growth of $\model'$.  $\qed$

\end{proof}

\end{document}